\documentclass{llncs}

\usepackage{amsmath, amssymb}
\usepackage[colorlinks=true, linkcolor = blue]{hyperref}
\usepackage{bbm}
\usepackage{url}
\usepackage{algorithmic}
\usepackage{algorithm}
\usepackage{verbatim}
\def\a{\alpha}
\def\b{\beta}

\def\g{\gamma}
\def\d{\delta}
\def\e{\epsilon}
\def\ve{\varepsilon}
\def\D{\Delta}

\def\l{\lambda}

\def\s{\sigma}
\def\brR{\bar{R}}
\def\beq{\begin {equation}}
\def\eeq{\end {equation}}
\def\beqar{\begin {eqnarray*}}
\def\eeqar{\end {eqnarray*}}

\DeclareMathOperator*{\argmax}{argmax}

\newcommand{\id}{\ensuremath{\mathbbm{1}}}
\newcommand{\eg}{\emph{e.g.}}
\newcommand{\ie}{\emph{i.e.}}
\newcommand{\etal}{\emph{et al}}

\newcommand{\expect}{\mathbb{E}}
\newcommand{\real}{\mathbb{R}}
\newcommand{\prob}{\mathbb{P}}

\newcommand{\bw}{\mathbf{w}}
\newcommand{\bv}{\mathbf{v}}

 \newcommand{\bb}{\mathbf{b}}
\newcommand{\bd}{\mathbf{d}}

\newcommand{\bp}{\mathbf{p}}
\newcommand{\bx}{\mathbf{x}}

\newcommand{\by}{\mathbf{y}}
\newcommand{\bY}{\mathbf{Y}}
\newcommand{\bI}{\mathbf{I}}
\newcommand{\bu}{\mathbf{u}}
\newcommand{\ba}{\mathbf{a}}
\newcommand{\bk}{\mathbf{k}}
\newcommand{\bK}{\mathbf{K}}

\newcommand{\blambda}{\mathbf{\lambda}}
\newcommand{\bmu}{\mathbf{\mu}}

\def\rI{\mathrm{I}}

\newtheorem{thm}{Theorem}%[section]
\newtheorem{mydef}[theorem]{Definition}%[section]
\newtheorem{assumption}[theorem]{Assumption}%[section]
\newtheorem{cor}{Corollary}
\newtheorem{lem}{Lemma}
\newtheorem{prop}{Proposition}
\newcommand{\knapsack}{{knapsack}}

\title{Privacy Auctions for Recommender Systems}

\author{Pranav Dandekar\inst{1}
\and
Nadia Fawaz\inst{2}
\and
Stratis Ioannidis\inst{2}
 }
 \institute{Stanford University \email{ppd@stanford.edu}
\and
Technicolor \email{\{nadia.fawaz, stratis.ioannidis\}@technicolor.com}
}

\begin{document}

\maketitle

\begin{abstract}
%We extend the study of markets for private data initiated by Ghosh and Roth~\cite{ghosh-roth:privacy-auction}.

We study a market for private data in which a data analyst publicly releases a statistic over a database of private information.
Individuals that own the data incur a cost for their loss of privacy proportional to the differential privacy guarantee given by the analyst at the time of the release.
The analyst incentivizes individuals by compensating them, giving rise to a  \emph{privacy auction}.
Motivated by recommender systems, the statistic we consider is a linear predictor function with publicly known weights.
The statistic can be viewed as a prediction of the unknown data of a new individual, based on the data of individuals in the database.
We formalize the trade-off between privacy and accuracy in this setting, and show that a simple  class of estimates achieves an order-optimal trade-off.  It thus suffices to focus on auction mechanisms that output such estimates.
We use this observation to design a truthful, individually rational, proportional-purchase mechanism under a fixed budget constraint. We show that our mechanism is 5-approximate in terms of accuracy compared to the optimal mechanism, and that no truthful mechanism can achieve a $2-\ve$ approximation, for any $\ve > 0$.

%We study a market for private data in which a data analyst publicly releases a statistic over a database of private information.
%Individuals that own the data incur a cost for their loss of privacy proportional to the differential privacy guarantee given by the analyst at the time of the release.
%The analyst incentivizes individuals by compensating them, giving rise to a  \emph{privacy auction}.
%Motivated by recommender systems and, more generally, prediction problems, the statistic we consider is a linear predictor function with publicly known weights.
%The statistic can be viewed as a prediction of the unknown data value of an individual outside the database based on the database entries.
%We formalize the trade-off between privacy and accuracy in this setting, and show that a simple  class of estimates achieves an order-optimal trade-off.  It thus suffices to focus on auction mechanisms that output such estimates.
%We use this observation to design a truthful, individually rational, proportional-purchase mechanism under a fixed budget constraint. We show that our mechanism is 5-approximate in terms of accuracy compared to the optimal mechanism, and that no truthful mechanism can achieve a $2-\ve$ approximation, for any $\ve > 0$.

\end{abstract}

%\category{F.0}{Theory of Computation}{General}

%\terms{Theory, Economics, Privacy}

%\keywords{mechanism design, privacy auction, inner product}

%\acmformat{Dandekar, P., Fawaz, N., and Ioannidis, S. 2012. Privacy Auctions for Inner Product Disclosures.}

%\begin{bottomstuff}
%Author's addresses: P. Dandekar, Department of Management Science \& Engineering, Stanford University. Email: \href{mailto:ppd@stanford.edu}5{ppd@stanford.edu}. Research done in part while interning at Technicolor, Palo Alto, CA.
%N. Fawaz, Technicolor, Palo Alto, CA. Email: \href{mailto:nadia.fawaz@technicolor.com}{nadia.fawaz@technicolor.com}.
%S. Ioannidis, Technicolor, Palo Alto, CA. Email: \href{mailto:stratis.ioannidis@technicolor.com}{stratis.ioannidis@technicolor.com}.
%\end{bottomstuff}

\section{Introduction}

Recommender systems are ubiquitous on the Internet,  lying at the heart of some of the most popular Internet services, including Netflix, Yahoo, and Amazon. These systems use algorithms to predict, \emph{e.g.}, a user's rating for a movie,  her propensity to  click on an advertisement or to purchase a product online.
By design, such prediction algorithms rely on access to large training datasets, typically comprising data from thousands (often millions) of individuals.
%The monetization of user data has become a commonplace---yet, controversial---aspect of the Internet economy. A large ecosystem of companies (\emph{e.g.}, BlueKai and Acxiom) collect and aggregate data about individuals %(both online and offline)
%that they subsequently trade with third parties, while popular Internet services (\emph{e.g.},  Google, Facebook, and Netflix) routinely  mine user data to personalize their services to their customers.
%One of the most ubiquitous examples of this personalization is recommender systems.
%These systems are powered by algorithms that predict, \emph{e.g.}, the rating a user gives to a movie, or propensity to purchase a product or click on an ad, and rely heavily on training prediction models over datasets comprising behavioral data of thousands (often millions) of individuals.
This large-scale collection %and aggregation
of user data  has raised serious privacy concerns among researchers and consumer advocacy groups. Privacy researchers have shown that access to seemingly non-sensitive data (\emph{e.g.}, movie ratings) can lead to the leakage of potentially sensitive information when combined with de-anonymization techniques~\cite{deanon}. Moreover, a spate of recent lawsuits \cite{lawsuit1,lawsuit2,lawsuit3} as well as behavioral studies \cite{tailored} have demonstrated the increasing reluctance of the public to allow the unfettered collection and monetization of user data.

As a result, researchers and advocacy groups have argued in favor of legislation protecting individuals, by ensuring they can ``opt-out'' from data collection if they so desire \cite{DNT2011}.
However, a widespread restriction on data collection would be detrimental to profits of the above companies.
One way to address this tension between the value of data and the users' need for privacy is through \emph{incentivization}.
In short, companies releasing an individual's data ought to appropriately compensate her for the violation of her privacy, thereby incentivizing her consent to the release.

We study the issue of user incentivization through \textit{privacy auctions}, as introduced by Ghosh and Roth \cite{ghosh-roth:privacy-auction}. In a privacy auction, a data analyst has access to a database $\bd\in \real^n$ of private data $d_i$,  $i=1,\ldots,n$, each corresponding to a different individual. This data may represent information that is to be protected, such as an individual's propensity to click on an ad or purchase a product, or the number of visits to a particular website.
  The analyst wishes to publicly release an estimate $\hat{s}(\bd)$ of a statistic $s(\bd)$ evaluated over the database.
In addition, each individual incurs a privacy cost $c_i$ upon the release of the estimate $\hat{s}(\bd)$, and must be appropriately compensated by the analyst for this loss of utility. The analyst has a budget, which limits the total compensation paid out. As such, given a budget and a statistic $s$, the analyst  must (a) solicit the costs of individuals $c_i$ and (b) determine the estimate $\hat{s}$ to release as well as the appropriate compensation to each individual.

%Ghosh and Roth propose doing so through the use
Ghosh and Roth employ \emph{differential privacy}~\cite{dwork:sensitivity} as a principled approach to quantifying the privacy cost $c_i$.
 %, provides a rigorous and quantifiable definition of the privacy loss incurred by an individual.
% It has been successfully applied to the design of privacy-preserving algorithms for a wide range of  statistical queries and machine learning operations (\eg, \cite{dwork:differentialprivacy,dwork:sensitivity,Dwork04privacy-preservingdatamining,McSherry:2009}).
% Informally, a randomized function $\hat{s}(\bd)$ is $\epsilon$-differentially private with respect to individual $i$ if changing the data $d_i$ alters the probability distribution of the function output by at most an $e^\epsilon$ factor.
% As such, the parameter $\epsilon$ quantifies the privacy guarantee the analyst provides to the individual: a small $\epsilon$ corresponds to better privacy since it guarantees that $\hat{s}(d)$ is essentially independent of $d_i$.
Informally, ensuring that $\hat{s}(\bd)$  is $\epsilon$-differentially private with respect to individual $i$ provides a guarantee on the privacy of this individual; a small $\epsilon$ corresponds to better privacy since it guarantees that $\hat{s}(\bd)$ is essentially independent of the individual's data $d_i$.
Privacy auctions incorporate this notion by assuming that each individual $i$ incurs a cost $c_i=c_i(\epsilon)$, that is a function of the privacy guarantee $\epsilon$ provided by the analyst.
\subsection{Our Contribution}
Motivated by recommender systems, we focus in this paper on a scenario where the statistic $s$ takes the form of a \emph{linear predictor}:
 \begin{align}\textstyle s(\bd) := \langle\bw, \bd\rangle = \sum_{i=1}^n w_i d_i,\label{def:s}\end{align}
 where $\bw\in \real^n$, is a publicly known vector of real (possibly negative) weights.
Intuitively, the public weights $w_i$ serve as measures of the similarity between each individual $i$ and a new individual, outside the database. The function $s(\bd)$ can then be interpreted as a prediction of the value $d$ for this new individual.

Linear predictors of the form~\eqref{def:s} include many well-studied methods of statistical inference, such as the $k$-nearest-neighbor method, the Nadaranya-Watson weighted average, ridge regression, as well as support vector machines. We provide a brief review of such methods in Section~\ref{sec:prediction}. Functions of the form~\eqref{def:s} are thus of particular interest in the context of recommender systems \cite{sarwar:grouplens,item-cf}, as well as other applications involving predictions (\eg, polling/surveys, marketing). %and  motivate our study of releases of form \re{def:s}.
In the sequel, we ignore the provenance of the public weights $\bw$, keeping in mind that any of these methods apply.
Our contributions are as follows:
\begin{enumerate}
\item{\bf Privacy-Accuracy Trade-off.} We characterize the accuracy of the estimate $\hat{s}$ in terms of the \emph{distortion} between the linear predictor $s$ and $\hat{s}$ defined as  $\d(s,\hat{s}) := \max_{\bd} \expect\left[|s(\bd) - \hat{s}(\bd)|^2\right]$, \ie,
 the maximum mean square error between $s(\bd)$ and $\hat{s}(\bd)$ over all databases  $\bd$.
 We define a \emph{privacy index} $\beta(\hat{s})$ that captures the amount of privacy an estimator $\hat{s}$ provides to individuals in the database.
 %Our definitions of privacy and distortion enable us to show that.
 We show that any estimator $\hat{s}$ with low distortion must also have a low privacy index (Theorem~\ref{thm:acc-lb}).
% A lower distortion corresponds to better accuracy. Interpreted as a worst-case mean squared error, $\delta(s,\hat{s})$
%is a natural metric to consider. It is similar in form to the notion of expected utility loss considered in \cite{grs:universal-privacy}.
\item {\bf Laplace Estimators Suffice.} %Another benefit of our definition of distortion is that when $\hat{s}$ is a \emph{Laplace estimator} \cite{dwork:sensitivity,dwork:differentialprivacy} (\ie, $\hat{s}$ uses noise drawn from a Laplace distribution to guarantee privacy), we can transform the problem of minimizing distortion into one that resembles the knapsack problem.
%However, in order to justify designing a privacy auction that outputs a Laplace estimator,  we must argue that among all possible estimators of the inner product, focusing on Laplace estimators suffices.
We  show that a special class of \emph{Laplace estimators} \cite{dwork:sensitivity,dwork:differentialprivacy} (\ie, estimators that use noise drawn from a Laplace distribution), which we call Discrete Canonical Laplace Estimator Functions (DCLEFs),  exhibits an order-optimal trade-off between privacy and distortion (Theorem~\ref{thm:acc-ub}).
This  allows us to restrict our focus on privacy auctions that output DCLEFs as estimators of the linear predictor $s$.
\item {\bf Truthful, 5-approximate Mechanism, and Lower bound.}
We design a \emph{truthful},  \emph{individually rational}, and \emph{budget feasible} mechanism that outputs a DCLEF as an estimator of the linear predictor (Theorem~\ref{thm:budget-constr}).
Our estimator's accuracy is a 5-approximation with respect to the DCLEF output by an optimal, individually rational, budget feasible mechanism.
We also prove a lower bound (Theorem~\ref{thm:hardness}): there is no truthful DCLEF mechanism that achieves an approximation ratio $2-\ve,$ for any $\ve > 0$.
\end{enumerate}
%The mechanism we design is a \emph{proportional-purchase} mechanism: it selects a subset of individuals and sets their differential privacy guarantee $\e_i$ as well as their payment $p_i$ to be proportional to their absolute weight $|w_i|$.
%It guarantees perfect privacy (\ie, $\e_i = 0$) to all other individuals.

In our analysis, we exploit the fact that   when $\hat{s}$ is a Laplace estimator minimizing distortion under a budget resembles the \knapsack{} problem.
As a result, the problem of designing a privacy auction that outputs a DCLEF $\hat{s}$ is similar in spirit to the \knapsack{} auction mechanism~\cite{budget-feasible-mechanisms}. However, our setting poses an additional challenge because
the privacy costs %$c_i$
 exhibit \emph{externalities}: the cost incurred by an individual %$i$
 is a function of %the differential privacy guarantee $\epsilon_i$, which in turn depends on
which other individuals are being compensated.
Despite the externalities in costs, we achieve the same approximation as the one known for the \knapsack{} auction mechanism \cite{budget-feasible-mechanisms}.
%Moreover, the approximation ratio is independent of input parameters, such as the size of the domain in which the database entries $d_i$ take values.

\subsection{Related Work}

{\bf Privacy of behavioral data.} Differentially-private algorithms have been developed for the release of several different kinds of online user behavioral data such as click-through rates and search-query frequencies \cite{KorolovaKMN09}, as well as movie ratings \cite{Mcsherry:2009}. As pointed out by McSherry and Mironov \cite{Mcsherry:2009}, the reason why the release of such data constitutes a privacy violation is not necessarily that, \emph{e.g.}, individuals perceive it as embarrassing, but that it renders them susceptible to \emph{linkage} and \emph{de-anonymization attacks} \cite{deanon}.
Such linkages could allow, for example, an attacker to piece together an individual's address stored in one database with his credit card number or social security number stored in another database. %As a result, disclosing such online behavioral data could have serious privacy implications,  and therefore i
It is therefore natural to attribute a loss of utility to the disclosure of such data.

%A natural question to ask is whether individuals care about the privacy of seemingly non-sensitive online behavioral data such as ratings, clicks and websites visited.
%However, such data can be linked to other publicly available data to create powerful de-anonymization attacks which link multiple online identities belonging to an individual \cite{deanon,Mcsherry:2009}.
%Such linkages could allow an attacker to piece together, for example, an individual's address stored in one database with his credit card number or social security number stored in another database.
%As a result, disclosing such online behavioral data could have serious privacy implications,  and therefore it is natural to attribute a loss of utility to the privacy loss associated with such disclosures.

{\bf Privacy auctions.} Quantifying the cost of privacy loss allows one to study privacy in the context of an economic transaction.
Ghosh and Roth initiate this study of privacy auctions in the setting where the data is binary and the statistic reported is the sum of bits, \emph{i.e.}, $d_i\in\{0,1\}$ and  $w_i = 1$ for all $i=1,\ldots,n$ \cite{ghosh-roth:privacy-auction}.
%Using a slightly different notion of distortion\footnote{Ghosh and Roth quantify the accuracy of the released estimator $\hat{s}$ through a tail bound, termed $k$-accuracy, which is weaker than the definition of distortion we consider here (see Lemma \ref{lem:acc-lb}).}, they show that to maximize an estimator's accuracy it suffices to consider estimators $\hat{s}$ that provide an identical, positive privacy guarantee $\epsilon>0$ to a subset of individuals, and perfect privacy ($\epsilon=0$) to the remaining individuals. As such, a privacy auction in this setup reduces to a \emph{multi-unit procurement auction}: the analyst needs only to decide which individuals to ``purchase'' privacy from; all such individuals will receive precisely the same privacy guarantee $\epsilon$.
%Ghosh and Roth design such a mechanism, and show it is optimal (\emph{i.e.}, maximizing accuracy) among all \emph{truthful}, \emph{individually-rational}, \emph{budget-feasible}, and \emph{envy-free} mechanisms; the latter implies that no individual would rather exchange its payment and privacy guarantee with another individual.
%, and coincides with our mechanism in the case where the data assume bit values and weights are all equal.
 %the latter buys an identical amount of privacy from each of a subset of individuals, each of which receive an identical payment.
%Their mechanism is budget-feasible, individually rational, truthful, and \emph{envy free}: no individual would rather exchange its payment and privacy guarantee with another individual.
Unfortunately, the Ghosh-Roth auction mechanism cannot be readily generalized to asymmetric statistics such as \eqref{def:s}, which, as discussed in Section~\ref{sec:prediction}, have numerous important applications including recommender systems. Our Theorems~\ref{thm:acc-lb} and \ref{thm:acc-ub}, which parallel the characterization of order-optimal estimators in \cite{ghosh-roth:privacy-auction}, imply that to produce an accurate estimate of $s$, the estimator $\hat{s}$ \emph{must provide different privacy guarantees to different individuals}. This is in contrast to the multi-unit procurement auction of \cite{ghosh-roth:privacy-auction}. In fact, as discussed the introduction, a privacy auction outputting a DCLEF $\hat{s}(\bd)$ has many similarities with a \knapsack{} auction mechanism \cite{budget-feasible-mechanisms}, with the additional challenge of externalities introduced by the Laplacian noise (see also  Section~\ref{sec:mechanism}). %However, in our setting the costs incurred by individuals are coupled through the Laplace noise that is added to ensure differential privacy. As a result our analysis is more involved, and our approximation ratio somewhat more surprising than that for the \knapsack{} auction mechanism.

% identifying what privacy guarantee to provide to each individual turns into a knapsack-type of problem. Indeed, to produce an accurate estimate of $s$, the mechanism must provide different privacy guarantees to different individuals.   %In contrast to multi-unit procurement auctions, when weights are unequal the ``privacy goods" sold by individuals are no longer identical. Informally,  to produce an accurate estimate of $s$, the mechanism must differentiate among individuals. In Section~\ref{sec:mechanism}, we show that a privacy auction outputting a DCLEF as an estimator of the inner product has many similarities with this \knapsack{} mechanism.

{\bf Privacy and truthfulness in mechanism design.} A series of interesting results follow an orthogonal direction, namely, on the connection between privacy and truthfulness %in mechanism design
when individuals have the ability to misreport their data. Starting with the work of McSherry and Talwar \cite{mcsherrytalwar} followed by Nissim \etal\  \cite{approximatemechanismdesign}, Xiao  \cite{xiao:privacy-truthfulness} and most recently Chen \etal\  \cite{chen:privacy-truthfulness}, these papers design mechanisms that are simultaneously truthful and privacy-preserving (using differential privacy or other closely related definitions of privacy).
As pointed out by Xiao \cite{xiao:privacy-truthfulness},  all these papers consider an \textit{unverified} database, \ie, the mechanism designer cannot verify the data reported by individuals and therefore must incentivize them to report truthfully.
Recent work on truthfully eliciting private data through a \emph{survey} \cite{roth-liggett,roth-schoenebeck} %, whereb individuals first decide whether to participate in the survey and then report their private data,
 also fall under the unverified database setting \cite{xiao:privacy-truthfulness}.
In contrast, our setting, as well as that of Ghosh and Roth, is that of a \textit{verified} database, in which individuals cannot lie about their data.
This setting is particularly relevant to the context of online behavioral data: information on clicks, websites visited and products purchased is collected and stored in real-time and cannot be retracted after the fact.

%Another setting that is closely related to the setting of an unverified database is truthfully eliciting private data through a survey \cite{roth-liggett,roth-schoenebeck} where individuals first decide whether to participate in the survey and then report their private data.
%As pointed out by Xiao \cite{xiao:privacy-truthfulness}, incentivizing participation is closely related to incentivizing truth-telling in a mechanism design setting.

{\bf Correlation between privacy costs and data values.} An implicit assumption in privacy auctions as introduced in \cite{ghosh-roth:privacy-auction} is that the privacy costs $c_i$ are \emph{not} correlated with the data values $d_i$. This might not be true if, \emph{e.g.}, the data represents the propensity of an individual to contract a disease. %A possible criticism of our setting is that if individuals' privacy costs $c_i$ are correlated with their data values $d_i$, then participants in the auction might learn something about each others' data values by observing the estimator and their respective payments.
%In fact,
Ghosh and Roth \cite{ghosh-roth:privacy-auction} show that when the privacy costs are correlated to the data no individually rational direct revelation mechanism can simultaneously achieve non-trivial accuracy and  differential privacy. %with respect to both the data values as well as the privacy costs.
As discussed in the beginning of this section, the privacy cost of the release of behavioral data is predominantly due to the risk of a linkage attack.  It is reasonable in many cases to assume that this risk (and hence the cost of privacy loss) is not correlated to, \emph{e.g.}, the user's movie ratings.
%Therefore, it is reasonable to assume that in the context of online behavioral data, privacy costs are uncorrelated with data values.
Nevertheless, due to its importance in other settings such as medical data, more recent privacy auction models aim at handling such correlation  \cite{roth-liggett,roth-schoenebeck,fleischer:privacy-auctions-correlated};
%Our work is the first studying releases of asymmetric functions;
we leave generalizing our results to such privacy auction models as future work.

\section{Preliminaries}\label{sec:prelim}
Let $[k]=\{1,\cdots, k\}$, for any integer $k>0$, and define $\rI:= [R_{\min},R_{\max}]\subset \real$ to be a bounded real interval.
Consider a database containing the information of $n>0$ individuals.
In particular, the database comprises a vector $\bd$, whose entries $d_i\in \rI$, $i \in [n],$  represent the private information of individual $i$.
Each entry $d_i$ is \emph{a priori} known to the database administrator, and therefore individuals do not have the ability to lie about their private data.
A data analyst with access to the database would like to publicly release an estimate of the statistic $s(\bd)$ of the form \eqref{def:s}, \emph{i.e.}
%\beq\label{def:s}\textstyle
$s(\bd)=\sum_{i\in [n]} w_id_i$,
%\eeq
for some publicly known weight vector $\bw=(w_1,\dotsc,w_n) \in \real^n$.
For any subset $H\subseteq [n]$, we define $w(H) := \sum_{i\in H} |w_i|$, and  denote by $W :=w([n])= \sum_{i=1}^{n} |w_i|$ the $\ell_1$ norm of vector $\bw$.
 We denote the length of interval $\rI$ by $\D := R_{\max} - R_{\min}$, and its midpoint by $\brR := (R_{\min}+R_{\max})/2$.
Without loss of generality, we assume that $\ w_i \ne 0$ for all $i\in [n]$; if  not, since entries for which $w_i = 0$ do not contribute to the linear predictor,  it suffices to consider the entries of $\bd$ for which $w_i \neq 0$.

%For any subset $H=\{i_1,\dotsc,i_{|H|}\}\subseteq [n]$, we define the restriction of $\bw$ to H as the vector $\bw_{H}=(w_{i_1},\dotsc,w_{i_{|H|}})$. We denote the $L_1$-norm of vector $\bw$ by $W:=\|\bw \|_1= \sum_{i=1}^{n} |w_i|$, and similarly, the $L_1$-norm of vector $\bw_{H}$ by $W_H := \|\bw_{H} \|_1=\sum_{i\in H} |w_i|$.

%, and we denote by $W :=w([n])$ the total weight of all bits.

 %Some additional notation: for a subset $H\subseteq [n]$, we define $w(H) := \sum_{i\in H} w_i$.

%An \emph{estimator function} for inner product is any function $\hat{s}: \{0,1\}^n\rightarrow \real_+$ that takes as input a database $\bd$ and outputs a real number as the estimate of the inner product of the bits in $\bd$. A ``good" estimator function is one that, informally speaking, is accurate and also provides a good differential privacy guarantee. We will make these notions more precise shortly.
\subsection{Differential Privacy and Distortion}
Similar to \cite{ghosh-roth:privacy-auction}, we use the following generalized definition of differential privacy:
\begin{mydef}{\bf (Differential Privacy).}
A (randomized) function $f: \rI^n\rightarrow \real^m$ is $(\e_1,\dotsc, \e_n)$-differentially private if for each individual $i\in[n]$  and for any pair of data vectors $\bd, \bd^{(i)}\in \rI^n$ differing in only their $i$-th entry, $\e_i$ is the smallest value such that
${\prob[f(\bd)\in S]} \le e^{\e_i}{\prob[f(\bd^{(i)})\in S]}$ for all $S\subset \real^m$.
\end{mydef}
This definition differs slightly from the usual definition of $\e$-differential privacy~\cite{dwork:differentialprivacy}, as the latter is stated in terms of the \emph{worst case} privacy across all individuals.
More specifically, according to the notation in~\cite{dwork:differentialprivacy}, an  $(\e_1,\dotsc, \e_n)$-differentially private function is $\e$-differentially private,  where  $\e=\max_{i} \e_i$.

Given a deterministic function $f$, a well-known method to provide $\e$-differential privacy is to add random noise drawn from a Laplace distribution to this function  \cite{dwork:differentialprivacy}.
This readily extends to $(\e_1,\dotsc,\e_n)$-differential privacy.
%In particular, define the sensitivity of $f$ w.r.t.~the $i$-th bit as follows:
%\begin{mydef}{\bf (Sensitivity).}
%The sensitivity, $S_i(f)$, of a function $f: \{0,1\}^n\rightarrow \real$ with respect to  bit $d_i$,  $i\in [n]$, is given by
%\[
%S_i(f) := \max_{\bd, \bd^{(i)}\in \{0,1\}^n} |f(\bd) - f(\bd^{(i)})|
%\]
%where $\bd,\bd^{(i)}$ differ only in their $i$-th bit.
%\end{mydef}
%Informally speaking, $S_i(f)$ captures the extent to which the $i$-th bit affects the value of $f$; if $S_i(f)$ is high, flipping the $i$-th bit results in a large change in the function value.
%This definition generalizes the definition of $L_1$-sensitivity $S(f)$ appearing in \cite{dwork:sensitivity}; in particular, $S(f)$ is simply $\max_i S_i(f)$.
%Using the generalized definition of sensitivity, it is a straightforward extension of \cite{dwork:differentialprivacy}  to show that adding Laplace noise to $f$ guarantees $(\e_1,\dotsc,\e_n)$ differential privacy:
 \begin{lem}[\cite{dwork:differentialprivacy}]\label{lem:sensitivity}
Consider a deterministic function $f:\rI^n\rightarrow \real$. Define $\hat{f}(\bd) := f(\bd) + Lap(\s)$, where $Lap(\s)$ is a random variable sampled from the Laplace distribution with parameter $\sigma$. Then, $\hat{f}$ is $(\e_1,\dotsc, \e_n)$-differentially private, where
% \[
$\e_i = {S_i(f)}/{\sigma}$,  and %\quad i\in[n], \text{ and}
%\]
%\[
$S_i(f) := \max_{\bd, \bd^{(i)}\in \rI^n} |f(\bd) - f(\bd^{(i)})|$, % $i\in[n]$,
%\]
 is the sensitivity of $f$ to the i-th entry $d_i$, $i\in[n]$.
%where $S_i(f)$ is the sensitivity of $f$ with respect $d_i$.
\end{lem}
Intuitively, the higher the variance $\sigma$ of the Laplace noise added to $f$, the smaller $\epsilon_i$, and hence, the better the privacy guarantee of  $\hat{f}$.
Moreover, for a fixed $\sigma$, entries $i$ with higher sensitivity $S_i(f)$ receive a worse privacy guarantee (higher~$\e_i)$.

There is a natural  tradeoff between the amount of noise added and the accuracy of the perturbed function $\hat{f}$. To capture this, we introduce the notion of \emph{distortion}  between two (possibly randomized) functions:
%The weight $w_i$ determines the extent to which bit $d_i$ affects $s(\bd)$. Since all weights are not identical in general, different bits contribute to the sum $s(\bd)$ differently. Therefore, estimating $s(\bd)$ while providing differential privacy guarantees requires us to differentiate between individuals (e.g. those with high weights versus those with low weights). The definitions of function sensitivity and differential privacy above give us a formal way to make this differentiation.

%We define \emph{distortion} between two (possibly randomized) functions as follows.
\begin{mydef}{\bf (Distortion).}\label{def:distortion}
Given two functions $f: \rI^n\rightarrow \real$ and $\hat{f}: \rI^n\rightarrow \real$, the \emph{distortion}, $\d(f,\hat{f})$,  between $f$ and $\hat{f}$ is given by
\[
\d(f,\hat{f}) := \max_{\bd\in\rI^n} \expect\left[|f(\bd) - \hat{f}(\bd)|^2\right].
\]
\end{mydef}

In our setup, the data analyst wishes to disclose an \emph{estimator function} $\hat{s}: \rI^n\rightarrow \real$ of the linear predictor $s$.
%that takes as input a database $\bd$ and outputs a real number as the estimate of the inner product of the bits in $\bd$. A ``good" estimator function is one that, informally speaking, is accurate and also provides a good differential privacy guarantee. We will make these notions more precise shortly.
% Given an estimator function $\hat{s}: \{0,1\}^n\rightarrow \real_+$, we measure the accuracy of $\hat{s}$ by $\d(s,\hat{s})$.
Intuitively, a good  estimator $\hat{s}$ should have a small distortion $\delta(s,\hat{s})$, while also providing good differential privacy guarantees.

\subsection{Privacy Auction Mechanisms}
Each individual $i\in [n]$ has an associated cost function $c_i:\real_+\rightarrow \real_+$,  which determines the cost $c_i(\e_i)$ incurred by $i$ when an $(\e_1,\dotsc,\e_n)$-differentially private estimate $\hat{s}$ is released by the analyst. As in \cite{ghosh-roth:privacy-auction}, we consider linear cost functions, \ie,  $c_i(\e) = v_i\e$, for all  $i\in [n].$ We refer to $v_i$ as the \emph{unit-cost} of individual $i$. The unit-costs $v_i$ are not \emph{a priori} known to the data analyst.  Without loss of generality, we assume throughout the paper that $v_1 \le \dotsc \le v_n$.

Given  a weight vector $\bw = (w_1, \dotsc, w_n)\in \real^n$, let $M_s$ be a mechanism compensating individuals in $[n]$ for their loss of privacy from the release of an estimate $\hat{s}$ of the linear predictor $s(\bd)$. Formally, $M_s$ takes as input a vector of reported unit-costs $\bv = (v_1, \dotsc, v_n)\in \real_+^n$ and a budget $B$, and outputs 
%a payment $p_i\in \real_+$ for every $i\in [n]$, and
%an estimator function $\hat{s}: \{0,1\}^n\rightarrow \real_+$.
%
\begin{enumerate}
\item
a payment $p_i\in \real_+$ for every $i\in [n]$, and
%\item
%For each $i\in [n]$, the differential privacy guarantee $\e_i(\bv, \bw)\in\real_+$.
\item
an estimator function $\hat{s}: I^n\rightarrow \real_+$.
\end{enumerate}

Assume that  the estimator $\hat{s}$ satisfies $(\e_1,\dotsc, \e_n)$-differential privacy.
A mechanism is \emph{budget feasible} if $\sum_{i\in [n]} p_i\leq B$, \ie, the payments made by the mechanism are within the budget $B$.
Moreover, a mechanism is \emph{individually rational} if for all $i\in [n],\ p_i \ge c_i(\e_i) = v_i\e_i$, \ie, payments made by the mechanism exceed the cost incurred by individuals.
Finally, a mechanism is \emph{truthful} if  for all $i\in [n],\ p_i(v_i,v_{-i})  - v_i\e_i(v_i,v_{-i}) \ge p_i(v'_i,v_{-i})  - v_i\e_i(v'_i,v_{-i})$, \ie, no individual can improve her utility by misreporting her private unit-cost.

\subsection{Outline of our approach}
We denote by $\d_{M_s}:= \d(s, \hat{s})$ the distortion between $s$ and the function output by the mechanism $M_s$. Ideally, a mechanism should output an estimator that has small distortion. However, the smaller the distortion, the higher the privacy violation and, hence, the more money the mechanism needs to spend. As such, the objective of this paper is to design a mechanism with minimal distortion, subject to the constraints of  truthfulness, individual rationality, and budget feasibility.

To address this question, in Section~\ref{sec:priv-acc-tradeoff}, we first  establish a privacy-distortion tradeoff for differentially-private estimators of the linear predictor. We then introduce a family of estimators,  Discrete Canonical Laplace Estimator Functions (DCLEFs), and show that they achieve a near-optimal privacy-distortion tradeoff.
This result allows us to limit our attention to DCLEF privacy auction mechanisms, \ie,  mechanisms that output a DCLEF $\hat{s}$. In Section~\ref{sec:mechanism}, we present a mechanism that is truthful, individually rational, and budget feasible, while also being near-optimal in terms of distortion.

\section{Privacy-Distortion Tradeoff and Laplace Estimators}\label{sec:priv-acc-tradeoff}

Recall that a good estimator should exhibit low distortion and simultaneously give good privacy guarantees. In this section, we  establish the privacy-distortion tradeoff for differentially-private estimators of the linear predictor. Moreover, we introduce a family of estimators that exhibits a near-optimal tradeoff between privacy and distortion. This will motivate our focus on privacy auction mechanisms that output estimators from this class in Section~\ref{sec:mechanism}.

\subsection{Privacy-Distortion Tradeoff}

%We first establish formally a trade-off between the privacy and distortion of an estimator function $\hat{s}$.
%We show that there exists a simple, natural class of estimators, namely DCLEFs,  that attains a near-optimal trade-off between privacy and distortion.
%This motivates our study of DCLEF auction mechanisms.

There exists a natural tension between privacy and distortion, as highlighted by the following two examples.

\noindent {\bf Example 1.} Consider the estimator $\hat{s} := \brR\sum_{i=1}^n w_i$, where recall that $\brR = (R_{\min}+R_{\max})/2$. This estimator guarantees perfect privacy (\ie, $\e_i = 0$), for all individuals. However, $\d(s,\hat{s}) =  (W\D)^2/4$.

\noindent {\bf Example 2.} Consider the estimator function $\hat{s} := \sum_{i=1}^n w_id_i$. In this case, $\d(s,\hat{s}) = 0$. However, $\e_i = \infty$ for all $i\in [n]$.

In order to formalize this tension between privacy and distortion, we define the \emph{privacy index} of an estimator as follows.
\begin{mydef}\label{def:priv-index}
Let $\hat{s}: \rI^n\rightarrow \real$ be any $(\e_1, \dotsc, \e_n)$-differentially private estimator function for the linear predictor. We define the \emph{privacy index}, $\b(\hat{s})$, of $\hat{s}$ as
\begin{align}\label{privindex}
\b(\hat{s}) := \max\left\{ w(H) : H \subseteq [n] \text{ and } \sum_{i\in H} \e_i < 1/2\right\}.
\end{align}
\end{mydef}
$\b(\hat{s})$ captures the weight of the individuals that have been guaranteed good privacy by $\hat{s}$. Next we characterize the impossibility of having an estimator with a low distortion but a high privacy index. Note that for Example 1, $\beta(\hat{s})=W$, \ie, the largest value possible, while for Example 2, $\beta(\hat{s})=0$. We stress that  the selection of 1/2 as an upper bound in \eqref{privindex} is arbitrary; Theorems \ref{thm:acc-lb} and \ref{thm:acc-ub} still hold if another value is used, though the constants involved will differ.

Our first main result, which is proved in Appendix~\ref{app:proofofacc-lb}, establishes a trade-off between the privacy index and the distortion of an estimator.
\begin{thm}[Trade-off between Privacy-index and Distortion]
\label{thm:acc-lb}
Let $0 < \a < 1$. Let $\hat{s}: \rI^n\rightarrow \real$ be an arbitrary estimator function for the linear predictor. If $\d(s,\hat{s}) \le (\a W \D)^2/48$ then $\b(\hat{s}) \le 2\a W$.
\end{thm}
In other words, if an estimator has low distortion, the weight of individuals with a good privacy guarantee (\ie, a small $\e_i$) can be at most an $\a$ fraction of $2W$.

%It is intuitive that such a result should hold; observe, for example, that the estimator in example 1 guarantees perfect privacy while achieving a very poor distortion whereas that in example 2 achieves zero distortion while giving no privacy.

%It is intuitive that a such a trade-off between privacy and distortion should exist.

\subsection{Laplace Estimator Functions}

Consider the following family of estimators for the linear predictor $\hat{s}: \rI^n\rightarrow \real$:
\beq\label{def:hats}
\hat{s}(\bd; \ba,\bx, \s) := \sum_{i=1}^n w_id_ix_i + \sum_{i=1}^n w_ia_i(1-x_i) + \text{Lap}(\s)
\eeq
where $x_i\in [0,1]$, and each $a_i\in \real$ is a constant independent of the data vector $\bd$. This function family is parameterized by $\bx, \ba$ and $\s$. The estimator $\hat{s}$ results from distorting $s$ in two ways: (a) a randomized distortion by the addition of the Laplace noise, and (b) a deterministic distortion through a linear interpolation between each entry $d_i$ and some constant $a_i$. Intuitively, the interpolation parameter $x_i$ determines the extent to which the estimate $\hat{s}$ depends on  entry $d_i$.
Using Lemma~\ref{lem:sensitivity} and the definition of distortion, it is easy to characterize the privacy and distortion properties of such estimators.

\begin{lem}\label{lem:privdistlap} Given $w_i$, $i\in [n]$, let $s(\bd)$ be the linear predictor given by \eqref{def:s}, and $\hat{s}$ an estimator of $s$ given by \eqref{def:hats}. Then,
\begin{enumerate}
\item
$\hat{s}$ is $(\e_1,\dotsc,\e_n)$-differentially private, where
$
\e_i = \frac{\D|w_i| \: x_i}{\s},
$ $i\in [n]$.
\item
The distortion satisfies
$\delta(s,\hat{s})  \geq \big(\frac{\D}{2}\sum_{i=1}^n |w_i|(1-x_i)\big)^2 +  2\s^2,$
with equality attained when $a_i=\bar{R}$, for all $i\in [n]$.
\end{enumerate}
\end{lem}
The proof of this lemma can be found in Appendix~\ref{app:proofs-3}. Note that the constants $a_i$ do not affect the differential privacy properties of $\hat{s}$. Moreover, among all estimators with given $\bx$, the distortion $\delta(s,\hat{s})$ is minimized when $a_i=\brR$ for all $i\in[n]$.  In other words, to minimize the distortion without affecting privacy, it is always preferable to interpolate between $d_i$ and $\brR$. This motivates us to define the family of Laplace estimator functions as follows.

\begin{mydef}
Given $w_i$, $i\in [n]$, the \emph{Laplace estimator function family} (LEF) for the linear predictor $s$ is  the set of functions $\hat{s}: \rI^n\rightarrow \real$, parameterized by $\bx$ and $\s$,  such that
\beq\label{eq:hats}
\hat{s}(\bd; \bx, \s) = \sum_{i=1}^n w_id_ix_i + \brR\sum_{i=1}^n w_i(1-x_i) + \text{Lap}(\s)
\eeq
%The function family is parameterized by $\bx$ and $\s$.
\end{mydef}

We call a LEF \emph{discrete} if $x_i\in\{0,1\}$.
%Consider a subclass of Laplace estimator functions where the noise is proportional to the ``residual weight" of individuals. We call this class \emph{canonical} Laplace estimator functions.
%\begin{mydef}
%The \emph{canonical Laplace estimator function family} for inner product $s$ is  defined to be the set of Laplace estimator functions  $\hat{s}: \{0,1\}^n\rightarrow \real_+$ for which
%\beq\label{eq:canon-sigma}
%\s = \s(\bx) := \sum_{i=1}^n w_i(1-x_i)
%\eeq
%Since $\s$ is a function of $\bx$, the function family is parameterized by $\bx$.
%\end{mydef}
%Observe that for a canonical Laplace estimator $\hat{s}(\bd; \bx)$,
%\beq\label{eq:clef-dist-expr}
%\d(s,\hat{s}) = \frac94\left(\sum_{i=1}^n w_i(1-x_i)\right)^2
%\eeq
%(since $\s = \s(\bx) = \sum_{i=1}^n w_i (1-x_i)$). We will make heavy use of canonical Laplace estimator functions in this paper.
Furthermore, we call a LEF \emph{canonical} if the Laplace noise added to the estimator has a parameter of the form
\beq\label{eq:canon-sigma}
\s = \s(\bx) := \D \sum_{i=1}^n |w_i| (1-x_i)
\eeq
Recall that $x_i$ controls the dependence of $\hat{s}$ on the entry $d_i$; thus, intuitively,  the standard deviation of the noise added in a canonical Laplace estimator is proportional to the ``residual weight'' of  data entries. % excluded from the deterministic part of the estimator.
Note that, by Lemma~\ref{lem:privdistlap}, the distortion of a canonical Laplace estimator $\hat{s}$ has the following simple form:
\beq\label{eq:clef-dist-expr}
\d(s,\hat{s}) = \frac94 \D^2 \big( \sum_{i=1}^n |w_i| (1-x_i)\big)^2 = \frac94 \D^2 \big( W - \sum_{i=1}^n |w_i| x_i\big)^2.
\eeq

%As we will see in the next section, discrete canonical Laplace estimator functions (DCLEF, in shorthand)
%exhibit a near-optimal tradeoff between privacy and distortion. This will motivate our focus on DCLEF
%privacy auction mechanisms, \emph{i.e.},  mechanisms that output a DCLEF $\hat{s}$.
%

Our next result establishes that there exists a discrete canonical Laplace estimator function (DCLEF) with a small distortion and a high privacy index.
\begin{thm}[DCLEFs suffice]
\label{thm:acc-ub}
Let $0 < \a < 1$. Let
\[
\hat{s}^* := \argmax_{\hat{s}: \d(s, \hat{s})\le (\a W \D)^2/48} \b(\hat{s})
\]
 be an estimator with the highest privacy index among all $\hat{s}$ for which $\d(s, \hat{s})\le (\a W \D)^2/48$. There exists a DCLEF $\hat{s}^{\circ}: \rI^n\rightarrow \real$ such that $\d(s,\hat{s}^{\circ}) \le (9/4) (\a W  {\D})^2$, and $\b(\hat{s}^{\circ}) \ge \frac12 \b(\hat{s}^*)$.
\end{thm}
In other words, there exists a DCLEF that is within a constant factor, in terms of both its distortion and its privacy index, from an optimal estimator $\hat{s}^*$.
Theorem \ref{thm:acc-ub} is proved in Appendix~\ref{app:proofofacc-ub} and has the following immediate corollary:
\begin{cor}Consider an arbitrary estimator $\hat{s}$ with distortion $\delta(s,\hat{s})<(W \D)^2/48$. Then, there exists a DCLEF $\hat{s}^\circ$ such that
$\delta(s,\hat{s}^{\circ}) \le 108 \delta(s,\hat{s})$ and
$\b(\hat{s}^{\circ}) \ge \frac12 \b(\hat{s}).$
\end{cor}
\begin{proof}
Apply Theorem \eqref{thm:acc-ub} with $\alpha = \sqrt{48\d(s,\hat{s}})/(W \D)$. In particular, for this $\alpha$ and $\hat{s}$ as in the theorem statement, we have that
$\hat{s}^* := \argmax_{\hat{s}': \d(s, \hat{s}')\le \d(s, \hat{s})} \b(\hat{s}')$, hence $\beta(\hat{s}^*)\geq \beta(\hat{s})$. Therefore, there exists a DCLEF $\hat{s}^\circ$ such that $\delta(s,\hat{s}^{\circ}) \le (9/4) (\a W  {\D})^2 \le 108 \delta(s,\hat{s})$, and $\b(\hat{s}^{\circ}) \ge \frac12 \b(\hat{s}^*) \ge \frac12 \beta(\hat{s})$.
\end{proof}
%
%Theorem~\ref{thm:acc-ub} shows that there exists a DCLEF that is at most a small constant factor away in terms of both distortion and privacy index compared to the estimator with the highest privacy index among all ``low-distortion" estimators.
%
Theorems~\ref{thm:acc-lb} and~\ref{thm:acc-ub} imply that, when searching for estimators with low distortion and high privacy index, it suffices (up to constant factors) to focus on DCLEFs.  %serve the same purpose as
Similar results were derived in \cite{ghosh-roth:privacy-auction} for estimators of unweighted sums of bits.%; namely,  they simplify the space of estimator functions we need to focus on.
%Specifically, they %We prove Theorem~\ref{thm:acc-lb} and Theorem~\ref{thm:acc-ub} in Section~\ref{sec:estimators}.

%There are additional reasons why DCLEFS are interesting to work with.
%First, observe that for a DCLEF $\hat{s}(\ \cdot\ ;\bx)$, the vector $\bx$ partitions $[n]$ into two sets---those who have been ``selected" (i.e. those for whom $x_i = 1$) and those who weren't.
%Let $H := \{i : x_i = 1\}$, i.e. $H$ is the set of individuals in $[n]$ ``selected" by $\hat{s}$.
%Second, for each $i\in H$,  the privacy guarantee $\e_i$ is proportional to $w_i$, i.e. $\e_i = w_i/\s(\bx)$ (follows from Corollary~\ref{cor:canon-diff-priv}), whereas for all $i\notin H,\ \e_i = 0$, i.e. the remaining individuals get perfect privacy. This has the intuitive appeal that individuals with greater weight (and therefore, a greater influence on the estimate) incur a greater loss of privacy and consequently receive larger payments.
%These reasons make DCLEFs easy to understand and explain to users whose privacy is being bought.

\section{Privacy Auction Mechanism}\label{sec:mechanism}
Motivated by Theorems~\ref{thm:acc-lb} and~\ref{thm:acc-ub}, we design a truthful, individually rational, budget-feasible DCLEF mechanism (\ie, a mechanism that outputs a DCLEF) and show that it is 5-approximate in terms of accuracy compared with the optimal, individually rational, budget-feasible DCLEF mechanism.
Note that a DCLEF is fully determined by the vector $\bx\in \{0,1\}^n$.
Therefore, we will simply refer to the output of the DCLEF mechanisms described below as $(\bx, \bp)$, as the latter characterize the released estimator  and the compensations to individuals.

\subsection{An Optimal DCLEF Mechanism}

Consider the problem of designing a DCLEF mechanism $M$ that is individually rational  and budget feasible (but not necessarily truthful), and minimizes $\d_M$.
Given a DCLEF $\hat{s}$,  define  $H(\hat{s}) := \{i : x_i = 1\}$ to be the set of individuals that receive non-zero differential privacy guarantees.
Eq.~\eqref{eq:clef-dist-expr} implies that  $\d(s,\hat{s}) = \frac94 \D^2( W - w(H(\hat{s})) )^2$.
Thus, minimizing $\d(s,\hat{s})$ is equivalent to maximizing $w(H(\hat{s}))$.
Let $(\bx_{opt},\bp_{opt})$ be an optimal solution to the following problem:
\begin{align}
\label{eq:budget-constr-prob}
\begin{split}
\text{maximize } \quad& S(\bx; \bw) = \sum_{i=1}^n |w_i| x_i  \\
\text{subject to:} \quad&  p_i \ge  v_i\e_i(\bx),\quad \forall i\in [n],\quad \text{ (individual rationality)}\\
& \sum_{i=1}^n p_i \le B\qquad\qquad\quad \text{ (budget feasibility)}\\
 & x_i \in \{0,1\},~ \forall i\in[n]\quad\text{ (discrete estimator function)}
\end{split}
\end{align}
where, by Lemma~\ref{lem:privdistlap} and \eqref{eq:canon-sigma}, \begin{align}\label{eq:ei}\e_i(\bx) = \frac{\D|w_i| x_i}{\s(\bx)} = \frac{|w_i| x_i }{\sum_i |w_i|(1-x_i)} \text{ (canonical property)}.\end{align}
A mechanism $M_{opt}$ that outputs $(\bx_{opt},\bp_{opt})$ will be an optimal, individually rational, budget feasible (but not necessarily truthful) DCLEF mechanism.
Let $OPT := S(\bx_{opt}; \bw)$ be the optimal objective value of \eqref{eq:budget-constr-prob}.
We use $OPT$ as the benchmark to which we compare the (truthful) mechanism we design below. Without loss of generality, we make the following assumption: % about the inputs to the mechanism.
\begin{assumption}\label{asmp:small-weights}
For all $i\in [n],\ |w_i|v_i/(W - |w_i|) \le B$.
\end{assumption}
\noindent Observe that if an individual $i$ violates this assumption, then $c_i(\e_i(\bx)) > B$ for any $\bx$ output by a  DCLEF mechanism that sets $x_i = 1$. In other words, no DCLEF mechanism (including $M_{opt}$) can compensate this individual within the analyst's budget and, hence,  will set $x_i = 0$. Therefore, it suffices to focus on the subset of individuals for whom the assumption holds.

\subsection{A Truthful DCLEF Mechanism}
To highlight the challenge behind designing a truthful DCLEF mechanism, observe that if the privacy guarantees were given by $\e_i(\bx) = x_i$ rather than \eqref{eq:ei},  the optimization problem \eqref{eq:budget-constr-prob} would be identical to the budget-constrained mechanism design problem for knapsack studied by Singer  \cite{budget-feasible-mechanisms}. In the reverse-auction setting of \cite{budget-feasible-mechanisms}, an auctioneer purchases items valued at fixed costs $v_i$ by the individuals that sell them.
Each item $i$ is worth $|w_i|$ to the auctioneer,  while the auctioneer's  budget is $B$.
The goal of the auctioneer is to maximize the total worth of the purchased set of items, \ie, $S(\bx;\bw)$. Singer presents a truthful mechanism that is 6-approximate with respect to $OPT$.
 However, in our setting, the privacy guarantees $\e_i(\bx)$ given by \eqref{eq:ei} introduce \emph{externalities} into the auction. In contrast to \cite{budget-feasible-mechanisms}, the $\e_i$'s couple the cost incurred by an individual $i$ to the weight of other individuals that are compensated  by the auction, making the mechanism design problem harder.
This difficulty is overcome by our mechanism, which we call FairInnerProduct, described in Algorithm~\ref{mech:FairInnerProduct}.

\begin{algorithm}[t]
\begin{algorithmic}
\STATE Let $k$ be the largest integer such that $\frac{B}{w([k])}\ge \frac{v_k}{W - w([k])}$.
\STATE Let $i^* := \argmax_{i\in[n]} |w_i|$.
\STATE Let $\hat{p}$ be as defined in \eqref{def:hatp}.
\IF{$|w_{i^*}| > \sum_{i\in[k]\setminus\{i^*\}} |w_i|$}
	\STATE Set $O = \{i^*\}$.
	\STATE Set $p_{i^*} = \hat{p}$ and $p_i = 0$ for all $i\ne i^*$.
\ELSE
	\STATE Set $O = [k]$.
	\STATE Pay each $i \in O,\ p_i = |w_i|\min\{\frac{B}{w([k])}, \frac{v_{k+1}}{W - w([k])}\}$, and for $i \notin O,\ p_i = 0$.
\ENDIF
\STATE Set $x_i = 1$ if $i\in O$ and $x_i = 0$ otherwise.
\end{algorithmic}
\caption{FairInnerProduct($\bv,\bw, B$)}
\label{mech:FairInnerProduct}
\end{algorithm}

The mechanism takes as input the budget $B$, the weight vector $\bw$, and the vector of unit-costs $\bv$, and outputs a set $O\subset [n]$, that receive $x_i=1$ in the DCLEF, as well as a set of payments for each individual in $O$. Our construction uses  a greedy approach similar to the Knapsack mechanism in \cite{budget-feasible-mechanisms}. In particular, it identifies users that are the ``cheapest'' to purchase. To ensure truthfulness, it compensates them within budget based on the unit-cost of the last individual that was not included in the set of compensated users. As in greedy solutions to \knapsack{}, this construction does not necessarily yield a constant approximation w.r.t.~OPT; for that, the mechanism needs to sometimes compensate only the user with the highest absolute weight $|w_i|$. In such cases, the payment of the user of the highest weight is selected so that she has no incentive to lie about here true unit cost.

Recall that $v_1 \le \dotsc \le v_n$. The mechanism defines  $i^* := \argmax_{i\in[n]} |w_i|$ as the individual with the largest $|w_i|$, and $k$ as the largest integer such that $\frac{B}{w([k])}\ge \frac{v_k}{W - w([k])}$. Subsequently, the mechanism either sets $x_i = 1$ for the first $k$ individuals, or, if $|w_{i^*}| > \sum_{i\in[k]\setminus\{i^*\}} |w_i|$, sets $x_{i^*} = 1$. In the former case, individuals  $i\in [k]$ are compensated \emph{in proportion to their absolute weights} $|w_i|$.  If, on the other hand, only $x_{i^*}=1$, the individual $i^*$ receives a payment $\hat{p}$ defined as follows:   Let %$S_{-i^*}$ be
\begin{align*}
S_{-i^*} :=& \Big\{t\!\in\! [n]\!\setminus\!\{i^*\} : \frac{B}{\sum_{i\in[t]\setminus\{i^*\}}\!\!\! |w_i|} \ge \frac{v_t}{W - \sum_{i\in[t]\setminus\{i^*\}}\!\! |w_i|} \text{ and }\!\!\!\!\!\!\!\!\! \sum_{i\in [t]\setminus\{i^*\}}\!\!\!\!\!\! |w_i| \ge |w_{i^*}| \Big\}.
\end{align*}
If $S_{-i^*} \ne \emptyset$, then let $r := \min\{i : i\in S_{-i^*}\}$. Define
\beq\label{def:hatp}
\hat{p} := \left\{
\begin{array}{rl}
B,& \text{ if } S_{-i^*} = \emptyset\\
\frac{|w_{i^*}|v_r}{W - |w_{i^*}|},& \text{ otherwise}
\end{array}
\right.
\eeq
The next theorem states that FairInnerProduct  has the properties we desire.
\begin{thm}\label{thm:budget-constr}
FairInnerProduct is truthful, individually rational and budget feasible. It is 5-approximate with respect to $OPT$. Further, it is 2-approximate when all weights are equal.
\end{thm}
The theorem is proved in Appendix~\ref{app:mechanism}.
We note that the truthfulness of the knapsack mechanism in \cite{budget-feasible-mechanisms} is established via Myerson's characterization of truthful single-parameter auctions (\emph{i.e.}, by showing that the allocation is monotone and the payments are threshold).
In contrast, because of the coupling of costs induced by the Laplace noise in DCLEFs, we are unable to use Myerson's characterization and,
instead, give a direct argument about truthfulness.

%Here is an informal description of the truthfulness argument: say an individual $i$ misreports his unit-cost $v_i$.  First, assume $i\ne i^*$.  If as a result of the misreport, the mechanism selects $\{i^*\}$, then clearly $i$ did not benefit from misreporting.  So assume that the mechanism selects a set $T$, and $i\in T$.  Then, we show that if $i < k$, then $T = [k]$, and $i$'s payment (and cost) remains unchanged because the selected set is unchanged.  Also, if $i \ge k$, then $T\setminus \{i\}\subseteq [k]$. Therefore, $i$'s payment is no more than his cost.  Next assume $i = i^*$.  If the mechanism selects $\{i\}$ in both cases (\ie, whether $i$ reports truthfully or not), then his payment remains the same.  So first assume that the mechanism selects $[k]$ when $i$ reports truthfully but selects $\{i\}$ when he misreports.  Then we show that there exists another individual $j$ such that $j\in S_{-i^*}$ and $v_j \le v_i$. So $i$'s payment $p_i \le |w_i|v_j/(W - w_i)\le c_i(\e_i)$.  Next assume that the mechanism selects $\{i\}$ when $i$ reports truthfully but selects a set $T$ when he misreports, and $i\in T$.  Then we show that $i\notin [k]$, and therefore $v_i \ge v_{k+1}$. Moreover, we show that $T\setminus\{i\}\subseteq [k]$ and therefore $p_i \le |w_i|v_{k+1}/(W - w(T))\le c_i(\e_i)$.

We prove a 5-approximation by using the optimal solution of the fractional relaxation of \eqref{eq:budget-constr-prob}.
This technique can also be used to show that the knapsack mechanism in \cite{budget-feasible-mechanisms} is 5-approximate instead of 6-approximate.
FairInnerProduct  generalizes the Ghosh-Roth mechanism; in the special case when all weights are equal FairInnerProduct reduces to the Ghosh-Roth mechanism, which, by Theorem~\ref{thm:budget-constr}, is 2-approximate with respect to $OPT$.
In fact, our next theorem, proved in Appendix~\ref{sec:hardness}, states that the approximation ratio of a truthful mechanism is lower-bounded by~2.
\begin{thm}[Hardness of Approximation]\label{thm:hardness}
For all  $\varepsilon>0$, there is no truthful, individually rational, budget feasible  DCLEF mechanism that is also $2-\varepsilon$-approximate with respect to $OPT$.
\end{thm}
Our benchmark $OPT$ is stricter than that used in \cite{ghosh-roth:privacy-auction}. In particular, Ghosh and Roth show that their mechanism is optimal among all truthful, individually rational, budget-feasible, and \emph{envy-free} mechanisms.
In fact, the example we use to show hardness of approximation is a uniform weight example, implying that the lower-bound also holds for uniform weight case.
Indeed, the mechanism in \cite{ghosh-roth:privacy-auction} is 2-approximate with respect to $OPT$, although it is optimal among individually rational, budget feasible mechanisms that are also truthful \emph{and} envy free.

\section{Discussion on Linear Predictors}
\label{sec:prediction}
%In our setup, we study the scenario where an analyst wishes to release a statistic $s(\bd)$ of the form \eqref{def:s}. In the case where the publicly known vector $\bw$ is non-negative, $s(\bd)$ has a clear interpretation as a weighted average over the data. Nevertheless,
As discussed in the introduction, a statistic $s(\bd)$ of the form \eqref{def:s} can be viewed as a \emph{linear predictor} and is thus of particular interest in the context of recommender systems. We elaborate on this interpretation in this section.
 Assume that each individual $i\in [n] =\{1,\ldots,n\}$ is endowed with a public vector $\by_i\in \real^m$, which includes $m$ publicly known features about this individual. These could be, for example, demographic information such as age, gender or zip code, that the individual discloses in a public online profile.  Note that, though features $\by_i$ are public, the data $d_i$ %(\emph{e.g.}, movie rating, clickthrough rate, \emph{etc.})
is perceived as private.

Let $\bY = [\by_i]_{i\in [n]}\in \real^{n \times m}$ be a matrix comprising public feature vectors. Consider a new individual, not belonging to the database, whose public feature profile is $\by\in \real^m$. Having access to $\bY$, $\bd$, and $\by$, the data analyst wishes to release a prediction for the unknown value $d$  for this new individual. Below, we give several examples where this prediction takes the form $ s(\bd) = \langle \bw, \bd\rangle  $, for some $\bw=\bw(\by,\bY)$. All examples are textbook  inference examples; we refer the interested reader to, for example, \cite{learning} for details.

\emph{$k$-Nearest Neighbors.}
In $k$-Nearest Neighbors prediction, the feature space $\real^m$ is endowed with a distance metric (\emph{e.g.}, the $\ell_2$ norm), and the predicted value is given by an average among the $k$ nearest neighbors of the feature vector $\by$ of the new individual. \emph{I.e.},
$s(\bd)  = \frac{1}{k} \sum_{i \in \mathcal{N}_k(\by)} d_i$
where $N_k(\by)\subset [n]$ comprises the $k$  individuals whose feature vectors $y_i$ are closest to $\by$. %This is a (somewhat trivial) statistic of the form \eqref{def:s}, where $w_i=1/k$ if $i\in N_k(\by)$ and zero otherwise.

\emph{Nadaranya-Watson Weighted Average.}
 The Nadaranya-Watson weighted average %(sometimes called ``collaborative filtering'' in recommender systems literature)
leverages all data in the database, weighing more highly data closer to $\by$. The general form of the prediction is
5$s(\bd) =\textstyle {\sum_{i=1}^n K(\by,\by_i) d_i}/\sum_{i=1}^n K(\by,\by_i) $
where the \emph{kernel} $K:\real^m\times\real^m\to\real_+$ is a function decreasing in the distance between its argument (\emph{e.g.}, $K(\by,\by')=e^{-\|\by-\by'\|^2}$).

\emph{Ridge Regression.}
%A simple way to predict the private value of the new user is using linear regression. Essentially, this amounts to assuming that
%the values $\bd$ are related to the features $\bY$ through a linear model, \emph{i.e.},
In ridge regression, the analyst first fits a linear model to the data, \emph{i.e.}, solves
%$$d_i \simeq \langle\by_i,\bb\rangle, \quad i\in [n]  $$
%for some unknown vector $\bb\in\real^m$. In this case, a common approach is to estimate $\bb$  by solving
 the optimization problem
\begin{align}\textstyle\min_{\bb\in\real^m} \sum_{i=1}^n \big(d_i-\langle \by_i,\bb\rangle\big)^2+\lambda \|\bb\|_2^2,  \label{linfit}\end{align}
where $\lambda\geq 0$ is a regularization parameter, enforcing that the vector $\bb$ takes small values. The prediction is then given by the inner product $\langle \by , \bb \rangle$.  The solution to \eqref{linfit} is given by
$\bb = (\bY^T\bY+\lambda \bI)^{-1}\bY^T \bd;$
 as such, the predicted value for a new user with feature vector $\by$ is given by
$s(\bd) = \langle \by, \bb \rangle = \by^T(\bY^T\bY+\lambda \bI)^{-1}\bY^T \bd$. %=\langle \bw,\bd \rangle. $$
%where $\bw = \bw(\by,\bY)=\by^T(\bY^T\bY+\lambda I)^{-1}\bY^T$.

\emph{Support Vector Machines.}
A more general regression model assumes that the private values $d_i$ can be expressed in terms of the public vectors $\by_i$ as a linear combination of a set of basis functions $h_\ell:\real^m\to \real$, $\ell=1,\ldots,L$, \emph{i.e.}, the analyst first solves
%$$d_i \simeq  \sum_{k=1}^K b_k h_k(\by_i). $$
%for, again, some unknown vector $\bb\in \real^K$. The latter is again obtained from the data $(\bY,\bd)$ by solving
the optimization problem
\begin{align}\label{minkern}\textstyle\min_{\bb\in\real^L} \sum_{i=1}^{n} \big(d_i -  \sum_{\ell=1}^L b_\ell h_\ell(\by_i)\big)^2 +\lambda \|\bb\|_2^2\end{align}
For $\by,\by'\in \real^m$, denote by $K(\by,\by')=\sum_{\ell=1}^Lh_\ell(\by)h_\ell(\by)$ the kernel of the space spanned by the basis functions. Let $\bK(\bY)=[K(\by_i,\by_j)]_{i,j\in [n]}\in \real^{n\times n}$ be the $n\times n$ matrix comprising the kernel values evaluated at each pair of feature vectors in the database, and $\bk(\by,\bY)=[K(\by,\by_i)]_{i\in[n]}\in \real^{n}$ the kernel values w.r.t. the new user. The solution to \eqref{minkern} yields a predicted value for the new individual of the form:
$s(\bd) = (\bk(\by,\bY))^T(\bK(Y)+\lambda \bI)^{-1} \bd $. %=\langle \bw(\by,\bY), \bd\rangle  $$
%for $\bw(\by,\bY) =(\bk(\by,\bY))^T(\bK(Y)+\lambda \bI)^{-1}$.

In all four examples, the prediction $s(\bd)$ is indeed of the form \eqref{def:s}. %, where the weights depend on the publicly known feature vectors of the individuals in the database as well as the new individual.
Note that the weights are non-negative in the first two examples, but may assume negative values in the latter two.

\section{Conclusion and Future Work}
We considered the setting of an auction, where a data analyst wishes to buy, from a set of $n$ individuals, the right to use their private data $d_i\in \real,\ i\in [n]$, in order to \emph{cheaply} obtain an \emph{accurate} estimate of a statistic.
Motivated by recommender systems and, more generally, prediction problems, the statistic we consider is a linear predictor with publicly known weights. 
The statistic can be viewed as a prediction of the unknown data of a new individual based on the database entries.
We formalized the trade-off between privacy and accuracy in this setting; we showed that obtaining an accurate estimate necessitates giving poor differential privacy guarantees to individuals whose cumulative weight is large.
We showed that DCLEF estimators achieve an order-optimal trade-off between privacy and accuracy, and,  consequently, it suffices to focus on DCLEF mechanisms. 
We use this observation to design a truthful, individually rational, budget feasible  mechanism under the constraint that the analyst has a fixed budget. 
Our mechanism can be viewed as a proportional-purchase mechanism, \ie, the privacy $\e_i$ guaranteed by the mechanism to individual $i$ is proportional to her weight $|w_i|$.
We show that our mechanism is 5-approximate in terms of accuracy compared to an optimal (possibly non-truthful) mechanism, and that no truthful mechanism can achieve a $2-\ve$ approximation, for any $\ve > 0$.

Our work is the first studying privacy auctions for asymmetric statistics, and can be extended in a number of directions.
An interesting direction to investigate is %designing of privacy auctions for more general statistics. 
characterizing the most general class of statistics for which truthful privacy auctions that achieve order-optimal accuracy can be designed.
An orthogonal direction is to study the release of asymmetric statistics in other settings such as (a) using a different notion of privacy, (b) allowing costs to be correlated with the data values, and (c) survey-type settings where individuals first decide whether to participate and then reveal their private data.

\bibliographystyle{splncs}
\bibliography{recsys_privacy_bib}

\appendix

\section{Proof of Theorem~\protect\lowercase{\ref{thm:acc-lb}} (Trade-off between Privacy-index and Distortion)}\label{app:proofofacc-lb}
By Definition~\ref{def:priv-index}, the privacy index  $\b(\hat{s})$ for an estimator $\hat{s}$ is the optimal objective value of the following optimization problem: maximize $\sum_{i=1}^n |w_i| x_i$ where $\sum_{i=1}^n \e_ix_i < \frac12$ and for all $i\in [n],\ x_i \in \{0,1\}$.

Interpreting  $|w_i|$ as the value, and $\e_i$ as the size of object $i$, the above problem can be viewed as a 0/1 knapsack problem where the size of the knapsack is 1/2.
%The value 1/2 is an artifact of the definition of $k$-accuracy as given in \cite{ghosh-roth:privacy-auction} (see below).
Assume for this proof, without loss of generality, that $\frac{\e_1}{|w_1|} \le \dotsc \le \frac{\e_n}{|w_n|}$. We define some notation that is needed in the proof. 
Let $h(\hat{s}):= \max\left\{j\in [n] : \frac{\e_j}{|w_j|}  < \frac1{2w([j])} \right\}$ if $\frac{\e_1}{|w_1|} <\frac1{2 |w_1|}$ and $h(\hat{s}):= 0$ otherwise.
%Define
%\[
%h(\hat{s}) :=
%\left\{
%\begin{array}{rl}
%0,& \text{ if } \frac{\e_1}{|w_1|} \ge \frac1{2 |w_1|}\\
% \max\left\{j\in [n] : \frac{\e_j}{|w_j|}  < \frac1{2w([j])} \right\},& \text{ otherwise}
% \end{array}
% \right.
%\]
Observe that $0 \le h(\hat{s}) \le n$. Next, define 
\[
\hat{i} := \argmax_{i\in [n]: \e_i < 1/2} |w_i|,
%\]
%Finally, define
%\[
\quad\text{and}\quad
H(\hat{s}) := \left\{\begin{array}{rl}
[h(\hat{s})],& \text{ if }w([h(\hat{s})]) \ge |w_{\hat{i}}|, \\
\{\hat{i}\},& \text{ otherwise}.
\end{array}
\right.
\]
The following then holds.
\begin{lem}\label{lem:approx-knapsack}
$2w(H(\hat{s})) \ge \b(\hat{s})$.
\end{lem}
\begin{proof}
 $H(\hat{s})$ is a  2-approximate greedy solution to the 0/1 knapsack problem  given by \cite[Section 2.4]{martello1990knapsack}. %, which is shown to be 2-approximate.
\end{proof}

Now we are ready to prove that if the distortion $\d(s, \hat{s})$ is small, then $w(H(\hat{s}))$ is also small,  which, together with Lemma~\ref{lem:approx-knapsack}, proves the theorem. In our proof, we make use of the notion of $k$-accuracy defined in \cite[Definition 2.6]{ghosh-roth:privacy-auction}.
%\begin{mydef}\label{def:k-accurate}
For $\hat{s}: \rI^n\rightarrow \real$, let % is $k_{\hat{s}}$-accurate for $k_{\hat{s}}$  given by
%\[
\begin{align}
\label{def:k-accurate}
k_{\hat{s}} := \min\left\{k\in \real_+: \forall \bd\in \rI^n, \prob[|s(\bd) - \hat{s}(\bd)| \ge k] \le \frac13 \right\}
\end{align}
%\end{mydef}

\begin{lem}\label{lem:acc-lb}
Let $0 < \a < 1$. If $w(H(\hat{s})) > \a W$ then $k_{\hat{s}} > \a W \Delta/4$.
\end{lem}
\begin{proof}
Assume for the sake of contradiction that $w(H(\hat{s})) > \a W$ and $k_{\hat{s}} \le \a W \Delta /4$. For a data vector $\bd$, let $z = s(\bd) = \sum_{i} w_id_i$ and $\hat{z} = \hat{s}(\bd)$. Also, let $S := \{y\in \real: |y - z| < k_{\hat{s}}\}$. Then, by \eqref{def:k-accurate}, $\prob[\hat{z}\in S] \ge 2/3$.

The set $H(\hat{s})$ can be partitioned as follows: $H(\hat{s})= H^+(\hat{s}) \cup H^-(\hat{s})$, with $H^+(\hat{s}) \cap H^-(\hat{s}) = \{ \emptyset \}$, where the disjoint subsets $H^+(\hat{s})$ and $ H^-(\hat{s})$ are defined by
\begin{equation}
\begin{split}
H^+(\hat{s})&= \{ i\in[n] : d_i \leq \brR \: \text{ and } \: w_i \leq 0 \} \cup \{ i\in[n] : d_i > \brR \: \text{ and } \:  w_i > 0 \},\\
H^-(\hat{s})&= \{ i\in[n] : d_i \leq \brR \: \text{ and } \: w_i > 0 \} \cup \{ i\in[n] : d_i > \brR \: \text{ and } \:  w_i \leq 0 \}.
\end{split}
\end{equation}
Then $w(H(\hat{s}))=w(H^+(\hat{s})) + w(H^-(\hat{s}))$. Thus, one of the subsets $H^+(\hat{s})$ and $ H^-(\hat{s})$ must have a total weight greater or equal to $w(H(\hat{s})) /2$. Without loss of generality, assume that $w(H^+(\hat{s})) \geq w(H(\hat{s})) /2 $.

Consider another data vector $\bd'$ where $d'_i = d_i$ if $i\in [n] \setminus H^+(\hat{s})$, while if $i\in H^+(\hat{s})$,
\begin{equation}\label{eq:bip}
d'_i = \left\{
\begin{array}{cc}
  d_i + \frac{\Delta}{2}, & \mbox{ if } d_i \leq \brR \: \text{ and } \: w_i \leq 0 \\
  d_i - \frac{\Delta}{2}, &  \mbox{ if } d_i > \brR \: \text{ and } \:  w_i > 0
\end{array}
\right.
\end{equation}
Let $z' := s(\bd') = \sum_{i=1}^n w_i d'_i$ and let $\hat{z}' = \hat{s}(\bd')$. Also, let $S' := \{y\in \real: |y - z'| < k_{\hat{s}}\}$. From eq. (\ref{eq:bip}), we have
\begin{equation}\label{eq:zz}
\begin{split}
|z - z'|
&\!=\! \big| \sum_{i\in H^+(\hat{s})}\!\!\!\!\!\! w_i (d_i - d'_i ) \big|\!=\! \big| \sum_{i\in H^+(\hat{s})} \!\!\!\!\!\!|w_i| \Delta /2 \big| %\\
%&
\!=\! \frac{\Delta}{2} w(H^+(\hat{s})) \geq \frac{\Delta}{4} w(H(\hat{s})) > \a \frac{\Delta}{4} W.
\end{split}
\end{equation}
Since $k_{\hat{s}} \le \a W \Delta/4$, eq. (\ref{eq:zz}) implies that $S$ and $S'$ are disjoint.

Since $\hat{s}$ is $(\e_1,\dotsc, \e_n)$-differentially private, and $\bd$ and $\bd'$ differ in exactly the entries in $H^+(\hat{s})$,
%\[
%\frac{\prob[\hat{z}\in S]}{\prob[\hat{z}'\in S]} \le \exp\left(\sum_{i\in H^+(\hat{s})} \e_i\right).
%\]
%In other words,
%\begin{align*}
$\prob[\hat{z}'\in S] \ge \exp\left(-\sum_{i\in H^+(\hat{s})} \e_i\right)\prob[\hat{z}\in S] \ge \exp\left(-\sum_{i\in H^+(\hat{s})} \e_i\right)\frac23.$
%\end{align*}
Note that $\sum_{i\in [h(\hat{s})]} \e_i < \sum_{i\in [h(\hat{s})]} \frac{|w_i|}{2w([h(\hat{s})])} = \frac12$, and also $\e_{\hat{i}} < 1/2$. Therefore, $\sum_{i\in H(\hat{s})} \e_i < 1/2$. Since $H^+ (\hat{s}) \subset H(\hat{s})$, we have $\sum_{i\in H^+(\hat{s})} \e_i \leq \sum_{i\in H(\hat{s})} \e_i < 1/2$. 

This implies
%\[
$\prob[\hat{z}'\in S]  \ge \exp\left(-\sum_{i\in H^+(\hat{s})} \e_i\right)\frac23 > \exp\left(-\frac12\right)\frac23 = \frac{2}{3\sqrt{e}} > \frac13.$
%\]
Given that $S$ and $S'$ are disjoint, $\prob[\hat{z}'\in S] > 1/3$ implies that $\prob[\hat{z}'\notin S'] > 1/3$, which contradicts the assumption that $k_{\hat{s}} \le \a W \Delta /4$.
\end{proof}

Next we relate $k_{\hat{s}}$-accuracy to the distortion $\delta(s,\hat{s})$: % in order to translate the $k$-accuracy lower bound obtained in Lemma~\ref{lem:acc-lb} into a distortion lower bound.

\begin{lem}\label{lem:dist-acc}
For $s(\bd)$ as defined in \eqref{def:s} and a function $\hat{s}: \rI^n\rightarrow \real$, %the accuracy of $\hat{s}$,
 $k_{\hat{s}} \le \sqrt{3\d(s,\hat{s})}$.
\end{lem}
\begin{proof}
Observe that for all $k \ge \sqrt{3\d(s,\hat{s})}$,
%\begin{align*}
$
\prob[|s(\bd) - \hat{s}(\bd)| \ge k] \le \prob[|s(\bd) - \hat{s}(\bd)| \ge \sqrt{3\d(s,\hat{s})}]%\\
 \le \frac{\expect[|s(\bd) - \hat{s}(\bd)|^2]}{3\d(s,\hat{s})}  \le \frac13
$
%\end{align*}
where the second step follows from Markov's inequality.
This implies $k_{\hat{s}} \le \sqrt{3\d(s,\hat{s})}$.
\end{proof}

\begin{cor} \label{cor:dist-lb}
If $w(H(\hat{s})) > \a W$ then $\d(s,\hat{s}) > (\a W \Delta )^2/48$.
\end{cor}
\begin{proof}
The corollary follows from Lemma~\ref{lem:acc-lb} and Lemma~\ref{lem:dist-acc}.
\end{proof}

Thus from Corollary~\ref{cor:dist-lb}, we have that if $\d(s,\hat{s}) \le (\a W \Delta)^2/48$, then $w(H(\hat{s})) \le \a W$. Since $w(H(\hat{s})) \ge \frac12\b(\hat{s})$ (from Lemma~\ref{lem:approx-knapsack}), it implies if $\d(s,\hat{s}) \le (\a W \Delta)^2/48$, then $\frac12\b(\hat{s}) \le \a W$. This concludes the proof of Theorem~\ref{thm:acc-lb}.\qed

\section{Proof of Lemma~\protect\lowercase{\ref{lem:privdistlap}}}\label{app:proofs-3}
For the first part of this lemma, observe that the sensitivity of $\sum_{i}w_i[x_id_i +(1-x_i)a_i]$ w.r.t.~$i$ is $S_i(\hat{s})=\D|w_i| x_i $. The differential privacy guarantee therefore follows from Lemma~\ref{lem:sensitivity}.

To obtain the lower bound on the distortion, observe that
substituting the expressions for $s$ and $\hat{s}$ in the expression for $\d(s,\hat{s})$, we get
\begin{align*}
%\d(s(\bd),\hat{s}(\bd,\ba,\bx, \s)) &= \max_{\bd\in\rI^n} \expect[\|s(\bd) - \hat{s}(\bd,\ba,\bx, \s)\|_2^2]\\
\d(s ,\hat{s}) &= \max_{\bd\in\rI^n} \expect[|s(\bd) - \hat{s}(\bd;\ba,\bx, \s)|^2]
\\
%&=  \max_{\bd\in\rI^n}\expect\left[\left(\sum_{i=1}^n w_id_i - \left(\sum_{i=1}^nw_id_ix_i + \sum_{i=1}^n w_ia_i(1-x_i)\right.\right.\right.\displaybreak[0]\\
%&\qquad\qquad\left.\left.\left. + z\right)\right)^2\right] \text{ (where $z \sim \text{Lap}(\s)$)}\displaybreak[0]\\
 &=  \max_{\bd\in\rI^n}\expect\big[\big(\sum_{i=1}^n w_id_i(1-x_i) - \sum_{i=1}^n w_ia_i(1-x_i) - z\big)^2\big]\text{ (where $z \sim \text{Lap}(\s)$)}\displaybreak[0]\\
% &=  \max_{\bd\in\rI^n}\expect\left[\left(\sum_{i=1}^n w_i(1-x_i)(d_i-a_i) - z\right)^2\right]\displaybreak[0]\\
% &= \max_{\bd\in\rI^n}\left(\sum_{i=1}^n w_i(1-x_i)(d_i-a_i)\right)^2 - 2\sum_{i=1}^n w_i(1-x_i)(d_i-a_i)\expect[z] + \expect[z^2]\displaybreak[0]\\
 &= \max_{\bd\in\rI^n}\big(\sum_{i=1}^n w_i(1-x_i)(d_i-a_i)\big)^2 + 2\s^2 \text{ (since $\expect[z] = 0; \expect[z^2] = 2\s^2$)}\displaybreak[0]\\
  &=2\s^2 +\max_{\bd\in\rI^n}\big(\sum_{i=1}^n \g_i(d_i-a_i)\big)^2 %\\ %\text{ (where $\g_i := w_i(1-x_i)$)}\\
  %&
= 2\s^2 +\big(\max_{\bd\in\rI^n}\big\vert\sum_{i=1}^n \g_i(d_i-a_i)\big\vert\big)^2
\end{align*}
Observe that %\[
$
\max_{\bd\in\rI^n}\vert f(\bd)\vert = \max\big\{\big\vert\max_{\bd\in\rI^n} f(\bd)\big\vert,\big\vert\min_{\bd\in\rI^n} f(\bd)\big\vert\big\}
$
%\]
for any continuous function $f:\rI^n\rightarrow \real$.  Therefore,
%\begin{equation}\label{eq:dhg}
\[
\begin{split}
\d(s ,\hat{s})
&= 2\s^2 +\big(\max\big\{\big\vert\max_{\bd\in\rI^n}\sum_{i=1}^n \g_i(d_i-a_i)\big\vert,\big\vert\min_{\bd\in\rI^n}\sum_{i=1}^n \g_i(d_i-a_i)\big\vert\big\}\big)^2\\
&= 2\s^2 +\big(\max\big\{\big\vert \gamma^{(+)}R_{\max}  + \gamma^{(-)}R_{\min} - \sum_{i=1}^n \gamma_i a_i \big\vert, \big\vert  \gamma^{(+)} R_{\min} + \gamma^{(-)}R_{\max} - \sum_{i=1}^n \gamma_i a_i \big\vert\big\}\big)^2,%\\
%&= 2\s^2 + (h(g(\ba)))^2.
\end{split}
\]
%\end{equation}
where $\gamma^{(+)} := \sum_{i:\gamma_i \geq 0}^n \gamma_i$,  %\text{and}\quad 
and $\gamma^{(-)} := \sum_{i: \gamma_i<0}^n \gamma_i.$
Observe that, for any $a,b,c,\in \real$, it is true that
$\max(|a-c|, |b-c|) \geq \frac{|a-b|}{2}$
with equality attained at $c= \frac{a+b}{2}$. Applying this for $a=\gamma^{(+)}R_{\max}  + \gamma^{(-)}R_{\min}$, $b= \gamma^{(+)} R_{\min} + \gamma^{(-)}R_{\max}$ and  $c= \sum_{i=1}^n \gamma_i a_i $ we get
%where the second equality holds by Lemma~\ref{lem:bmin-bmax}, and the third equality holds by the definition of $h$ and $g$ in Lemma~\ref{lem:hg}. Minizing (\ref{eq:dhg}) over $\ba \in \real^n$ yields
%\begin{equation}
%\begin{split}
%\min_{\ba \in \real^n}\d(s ,\hat{s}) = \d(s(\bd),\hat{s}(\bd;\ba^*,\bx, \s))= 2\s^2 + \left(\frac{\Delta}{2} \sum_{i=1}^n |w_i| (1-x_i)\right)^2,
%\begin{align*}
$
\min_{\ba \in \real^n}\d(s ,\hat{s}) \geq  2\s^2 + \frac{(\gamma^{+}-\gamma^{-})(R_{\max}-R_{\min})}{2} = 2\s^2 + \big(\frac{\Delta}{2} \sum_{i=1}^n |w_i| (1-x_i)\big)^2,
$
%\end{align*}
%\end{split}
%\end{equation}
with equality attained when $\sum_{i}\gamma_ia_i={(\gamma^{+}+\gamma^{-})(R_{\max}+R_{\min})}/{2}= \sum_{i}\gamma_i \bar{R},$
which holds for $a_i=\bar{R}$.
%where we recall that $\ba^*$ is defined by $a^*_i = \bar{R}$, $\forall i \in [n]$.
%This completes the proof of Lemma~\ref{lem:privdistlap}.
%\qedhere

%Next, observe that since $\g_i \ge 0$ for all $i$, $\sum_{i=1}^n \g_i(d_i-a_i)$ is maximized when $d_i = 1$ for all $i$, and $\sum_{i=1}^n \g_i(d_i-a_i)$ is minimized when $d_i = 0$ for all $i$. Therefore,
%\begin{align*}
%\d(s(\bd),\hat{s}(\bd,\ba,\s)) &= 2\s^2 + \left(\max\left\{\bigg\vert\sum_{i=1}^n \g_i(1-a_i)\bigg\vert,\bigg\vert\sum_{i=1}^n \g_ia_i\bigg\vert\right\}\right)^2
%\end{align*}

%Let $g:\real^n\rightarrow \real$ be defined as $g(\bx) = \sum_{i=1}^n \g_ix_i$. Let $h:\real \rightarrow \real$ be defined as $h(x) = \max\{|x|,|\r - x|\}$, where $\r := \sum_i \g_i$. Then, $\min_xh(x) = \r/2$. Therefore,
%\begin{align*}
%\min_{\ba}  \d(s(\bd),\hat{s}(\bd,\ba,\s)) &= 2\s^2 + \min_{\ba} \left(\max\left\{\bigg\vert\sum_{i=1}^n \g_i(1-a_i)\bigg\vert,\bigg\vert\sum_{i=1}^n \g_ia_i\bigg\vert\right\}\right)^2\\
%&= 2\s^2 + \min_{\ba} \left(\max\left\{\vert \r - g(\ba)\vert,\vert g(\ba)\vert\right\}\right)^2\\
%&= 2\s^2 + \min_{\ba} \left(h(g(\ba)\right)^2\\
%&\ge 2\s^2 + \min_x (h(x))^2 \text{ (since $\text{Range}(g)\subseteq \real$)} \\
%&= 2\s^2 + \frac{\r^2}{4}
%\end{align*}
%Also,
%\[
% \d(s(\bd),\hat{s}(\bd,\ba,\s)) = 2\s^2 + \frac{\r^2}{4}
%\]
%when $a_i = 1/2$ for all $i$. This completes the proof.%\qedhere
%\end{proof}

\section{Proof of Theorem~\protect\lowercase{\ref{thm:acc-ub}} (DCLEFs Suffice)}\label{app:proofofacc-ub}
Consider the function $\hat{s}^{\circ}(\bd) := \sum_{i\notin H^{\circ}} w_id_i + \brR \sum_{i\in H^{\circ}} w_i + Lap(w(H^{\circ}))$, where $H^{\circ}$ is defined as
%\[
$
H^{\circ} := \argmax\left\{w(H) : H\subseteq [n] \text{ and } w(H) \le \a W\right\}.
$
%\]
We can write $\hat{s}^{\circ}$ as $\hat{s}^{\circ}(\bd; \bx) := \sum_{i=1}^n w_id_ix_i + \brR \sum_{i=1}^n w_i(1-x_i) + \text{Lap}(w(H^{\circ}))$, where $x_i = 0$ for all $i\in H^{\circ}$ and $x_i = 1$ otherwise. Observe that $\hat{s}^{\circ}$ is a DCLEF and
%\begin{align*}
$\d(s,\hat{s}^{\circ}) \stackrel{\text{Lem.~\ref{lem:privdistlap}}}{=} \frac94 {\D}^2 \left(\sum_{i=1}^n |w_i| (1-x_i)\right)^2 %\text{ (from Lemma~\ref{lem:privdistlap})}\\
%&
= \frac94 {\D}^2 \left(w(H^{\circ})\right)^2 \le \frac94(\a W \D)^2.$
%\end{align*}
Since $\d(s,\hat{s}^*) \le (\a W \D)^2/48$, it follows from Lemma~\ref{lem:dist-acc} that $k_{\hat{s}^*} \le \a W \D/4$. Then, it follows from Lemma~\ref{lem:acc-lb} that $w(H(\hat{s}^*)) \le \a W$, where $H(\hat{s}^*)$ is as defined in the proof of Theorem~\ref{thm:acc-lb}. Further, it follows that $w(H^{\circ}) \ge w(H(\hat{s}^*)) \ge \frac12 \b(\hat{s}^*)$ (the first inequality follows by definition of $H^{\circ}$ and the fact that $w(H(\hat{s}^*)) \le \a W$, and the second from Lemma~\ref{lem:approx-knapsack}). Since $\b(\hat{s}^{\circ}) \ge w(H^{\circ})$, it follows that $\b(\hat{s}^{\circ}) \ge \frac12 \b(\hat{s}^*)$.

\section{Proof of Theorem~\protect\lowercase{\ref{thm:budget-constr}}}\label{app:mechanism}

\subsection{Truthfulness, Individual Rationality, and Budget Feasibility}
In this section, we prove that FairInnerProduct is truthful, individually rational, and budget feasible. %We first define two sets $S_1$ and $S_2$ that we use throughout this proof: %Define
%\begin{align*}
%S_1 &:= \left\{t\in [n]\setminus\{i^*\} : \frac{B}{\sum_{i\in[t]\setminus\{i^*\}} |w_i|} \ge \frac{v_t}{W - \sum_{i\in[t]\setminus\{i^*\}} |w_i|}\right\};\\\ S_2 &:= \left\{t\in [n]\setminus\{i^*\} : \sum_{i\in [t]\setminus\{i^*\}} |w_i| \ge |w_{i^*}| \right\}.
%\end{align*}
We first define $$S_1 := \left\{t\in [n]\setminus\{i^*\} : \frac{B}{\sum_{i\in[t]\setminus\{i^*\}} |w_i|} \ge \frac{v_t}{W - \sum_{i\in[t]\setminus\{i^*\}} |w_i|}\right\}$$ and $$S_2 := \left\{t\in [n]\setminus\{i^*\} : \sum_{i\in [t]\setminus\{i^*\}} |w_i| \ge |w_{i^*}| \right\}.$$
Observe that $S_{-i^*} = S_1 \cap S_2$.
\begin{prop}
FairInnerProduct is budget feasible.
\end{prop}
\begin{proof}
When $O = \{i^*\}$ and $\hat{p} = B$,  the mechanism is trivially budget feasible. If $\hat{p} = \frac{|w_{i^*}|v_r}{W - |w_{i^*}|}$ then observe that since $r\in S_{-i^*}$, this implies $r\in S_1$ and $r\in S_2$. Therefore,
%\[
$\hat{p} = \frac{|w_{i^*}|v_r}{W - |w_{i^*}|} \le \frac{|w_{i^*}|v_r}{W - \sum_{i\in [r]\setminus\{i^*\}} |w_i|} \le \frac{|w_{i^*}|B}{\sum_{i\in[r]\setminus\{i^*\}} |w_i|} \le B$
%\]
where the second inequality holds because $r\in S_1$ and the last inequality because $r\in S_2$. When $O = [k]$, the sum of the payments made by the mechanism is given by
%\[
$\sum_{i\le k} p_i \le \sum_{i\le k} |w_i|\frac{B}{w([k])} = \frac{B}{w([k])}\sum_{i\le k} |w_i| = B$. %\quad\qed
%\]
\end{proof}

\begin{prop}\label{prop:empty-si*}
If $i^* > k+1$ and $|w_{i^*}| > \sum_{i\in[k]\setminus\{i^*\}} |w_i|$, then $S_{-i^*} = \emptyset$.
\end{prop}
\begin{proof}
Observe that if $i^* > k+1$ and $|w_{i^*}| > \sum_{i\in[k]\setminus\{i^*\}} |w_i|$, then $S_1 = [k]$ and $S_2 \cap [k] = \emptyset$.
\end{proof}

\begin{prop}\label{prop:r>i}
If  $|w_{i^*}| > \sum_{i\in[k]\setminus\{i^*\}} |w_i|$ and $S_{-i^*} \ne \emptyset$, then $r > i^*$.
\end{prop}
\begin{proof}
From Proposition~\ref{prop:empty-si*}, $S_{-i^*} \ne \emptyset$ implies either $i^* \le k+1$ or $|w_{i^*}| \le \sum_{i\in[k]\setminus\{i^*\}} w_i$. Since the latter is false, it must be that $i^* \le k+1$. In that case, $S_2 \cap [k] = \emptyset$. Therefore $r > k$. If $i^* = k+1$, then for all $j\in  S_2,\ j\ge k+2$. Therefore $r \ge k+2$.
\end{proof}

\begin{prop}
FairInnerProduct is individually rational.
\end{prop}
\begin{proof}
We divide the proof into two cases:
%\begin{itemize}
\\\textbf{Case I:} $O =[k]$. We know that $B/w([k]) \ge v_k/(W - w([k]))$ (by construction) and $v_{k+1} \ge v_k$ (by definition). Therefore, for all $i\le k$,
%\[
$p_i \ge \frac{|w_i|v_k}{W - w([k])} \ge \frac{|w_i|v_i}{W - w([k])} = c_i(\e_i).$
%\]
\\\textbf{Case II:} $O = \{i^*\}$. If $p_{i^*} = B$, then the mechanism is individually rational by Assumption~\ref{asmp:small-weights}. If $p_{i^*} = \frac{|w_{i^*}|v_r}{W - |w_{i^*}|}$, then, by Proposition~\ref{prop:r>i}, $v_r\ge v_{i^*}$ and therefore the mechanism is individually rational.
\qed
%\end{itemize}
\end{proof}

\begin{prop}
FairInnerProduct is dominant-strategy truthful.
\end{prop}
\begin{proof}
Fix any $\bv$ and assume that user $i$ reports a value $z\neq v_i$, while the remaining values $\bv_{-i}$ remain the same. Let $\bu$ be the resulting vector of values, \emph{i.e.},
%$$u_j=\begin{cases}v_j, &j\neq i \\z,& j=i.\end{cases}$$
$u_i=z$ and $u_j=v_j$,  for $j\neq i.$ 
The vector $\bu$ induces a new ordering of the users in terms of their reported values $u_i$, $i\in[n]$; let $\pi:[n]\to [n]$ be the permutation indicating the position of users under the new ordering. That is, $\pi$ is 1-1 and onto such that
if $u_{j}< u_{j'}$ then $\pi(j)<\pi(j')$, for all $j,j'\in [n]$.
For given $j\in[n]$, we denote the set of users preceding $j$ under this ordering by $P_j = \{j': \pi(j')\leq \pi(j)\}$. Note that al $j'\in P_j$ satisfy $u_j'\leq u_j$.
%Let for $j\in[n]$, let
%$$ w_{\pi}([j]) = \sum_{\ell: \pi(\ell)\leq \pi(j)}w_\ell. $$
Observe that if $z>v_i$ then
\begin{align}\label{wplzup}
w(P_j) = \begin{cases}
w([j]),& \text{for all }j<i\\
w([j])-|w_i|, &\text{for all }j>i\text{ s.t.}~\pi(j)<\pi(i)\\
w([i])+w(\{\ell:\ell>i\land\pi(\ell)<\pi(i)\}), &\text{for }j=i\\
w([i]), &\text{for all}j>i\text{ s.t.}~\pi(j)>\pi(i)
\end{cases}
\end{align}
while if $z<v_i$ then
\begin{align}\label{wplzdown}
w(P_j) = \begin{cases}
w([j]),& \text{for all }j<i \text{ s.t.}~\pi(j)<\pi(i)
\\
w([i])-w(\{\ell:\ell<i\land\pi(\ell)>\pi(i)\}), &\text{for }j=i\\
w([j])+|w_i|, &\text{for all }j<i\text{ s.t.}~\pi(j)>\pi(i)\\
w([i]), &\text{for all }j>i\end{cases}
\end{align}
Let
$M_{\pi} = \left\{j\in[n]:\frac{B}{w(P_j)}\geq \frac{u_{j}}{W-w(P_j)} \right\} $
  where $W=w([n])$. Then, by \eqref{wplzup}, if $z>v_i$ then $w(P_j)\leq w([j])$ for $j\neq i$ while $w(P_i)\geq w([i])$. As a result, if $z>v_i$, then
\begin{subequations}
\label{mpiup}
\begin{align}
\text{for } j\neq i, \text{if } j\in [k],\ &\text{ then } j\in M_{\pi}\\
\text{if }i\notin[k],\ &\text{ then }i\notin M_{\pi}
\end{align}
\end{subequations}
Similarly, from \eqref{wplzdown}, if $z<v_i$, then
\begin{subequations}
\label{mpidown}
\begin{align}
\text{for } j\neq i, \text{if } j\notin [k],\ &\text{ then } j\notin M_{\pi}\\
\text{if }i\in[k],\ &\text{ then }i\in M_{\pi}
\end{align}
\end{subequations}

Observe that, given the value vector $\bu$, the mechanism will output
$O_{\pi} = \{i^*\}$, if $|w_i^*| > w(M_{\pi}\setminus \{i^*\})$, and $O_{\pi} = M_{\pi}$ otherwise.
%$$O_{\pi} = \begin{cases}\{|w_i^*|\}, &\text{if }|w_i^*| > w(M_{\pi}\setminus \{i^*\})\\
%M_{\pi}, &\text{o.w}.\end{cases}
% $$
If $O_{\pi}=M_{\pi}$, users $j\in M_{\pi}$ are compensated by $p_j=w_j \min \left\{\frac{B}{w(M_{\pi})}, \frac{\min_{\ell:\ell\notin M_{\pi}} u_\ell}{W-w(M_{\pi})}\right\}.$ If $O_{\pi}=\{i^*\}$, the latter is compensated by $\hat{p}$ given by \eqref{def:hatp}.
We consider the following cases:\\
\textbf{Case I:} $O_{\pi}=M_{\pi}$. If $i\notin M_{\pi}$, then $p_i=\epsilon_i=0$, so since FairInnerProduct is individually rational, $i$ has no incentive to report $z$. Suppose thus that $i\in M_{\pi}$.
We consider the following subcases:\\
\textbf{Case I(a)}: $i\notin[k]$. Then $v_i\geq v_{k+1}$. Since $i\in M_{\pi}$ but $i\notin[k]$, \eqref{mpiup} implies that $z<v_i$.  By \eqref{mpidown}  $k+1\notin M_{\pi}$. Thus $p_i\leq |w_i| v_{k+1}/w(M_{\pi})\leq |w_i| v_i/w(M_{\pi})$.\\
\textbf{Case I(b)}: $i\in[k]$.  We will first show that $M_{\pi}\setminus [k] = \emptyset$. Suppose, for the sake of contradiction, that $M_{\pi}\setminus [k]\neq \emptyset$. Then $M_\pi\setminus[k]$ must contain an element different than $i$; this, along with \eqref{mpidown} implies that $z>v_i$. If $\pi(i)<\pi(k+1)$, then by \eqref{wplzup} $w (P_j) =w([j])$ and $j\notin M_{\pi}$ for all $j\geq k+1$, which contradicts that $M_\pi\setminus [k]$ is non-empty. Hence, $\pi(i)>\pi(k+1)$; this however implies that
$w(P_i)\geq w([k+1])$, by \eqref{wplzup}, and that $z\geq v_{k+1}$. Thus
$\frac{B}{w(P_i)} \leq \frac{B}{ w([k+1]) } < \frac{v_{k+1}}{W-w([k+1])}\leq \frac{z}{W-w(P_i)},$ so $i\notin M_{\pi}$, a contradiction.
Hence $M_{\pi}\setminus [k]=\emptyset$.

Next we will show that the original output $O = [k]$. Suppose, for the sake of contradiction, that $O=\{i^*\}$. Then $|w_{i^*}|> w([k]\setminus \{i^*\})$ while $|w_{i^*}|\leq w(M_{\pi}\setminus \{i^*\}).$ Thus, $M_{\pi}\setminus [k]\neq \emptyset$, a contradiction.  Thus, $O=[k]$.

If $O_{\pi}=M_{\pi}=[k]$, then since $O=[k]$, user $i$ receives the same payoff, so it has no incentive to report $z$. Suppose that $M_{\pi}\neq [k]$. Since $M_{\pi}\setminus [k]=\emptyset$, it must be that $[k]\setminus M_{\pi}\neq \emptyset$. By \eqref{mpiup}, this implies $z<v_i$. If $i<k$, \eqref{wplzdown} implies that $k\in M_{\pi}$ and so do all $j$ s.t.~$\pi(j)<\pi(k)$. Thus, $[k]=M_{\pi}$, a contradiction. If $i=k$ and $z<v_i$, then it is possible that $j\notin M_{\pi}$ for some $j<k$. Thus,
%\[
$p_i \le \frac{|w_i| v_k}{w(M_{\pi})} = \frac{|w_i| v_i}{w(M_{\pi})}$
%\]
 and so $i$ has no incentive to report $z$.

\textbf{Case II}. $O_{\pi}=\{i^*\}$. If $i\neq i^*$, then $i$'s payoff is obviously zero, so it has no incentive to report $z$. Suppose thus that $i=i^*$.
We consider the following two subcases. \\
\textbf{Case II(a).} $O=\{i^*\}$. Observe that $S_{-i^*}$ and $\hat{p}$ do not depend on $v_{i^*}$. Thus, since $O=\{i^*\}$,  $i$ receives the same payment $\hat{p}$, so it has no incentive to misreport its value.\\
\textbf{Case II(b)} $O=[k]$. Then $|w_{i^*}|\leq w([k]\setminus \{i^*\})$ while $|w_{i^*}|> w(M_{\pi}\setminus \{i^*\}).$ Thus, $[k]\setminus M_{\pi}$ must contain an element different than $i^*$. From \eqref{mpiup}, this implies that $z<v_i$. If $i<k$,  \eqref{wplzdown} implies that $k\in M_{\pi}$ and so do all $j$ s.t.~$\pi(j)<\pi(k)$. Thus, $[k]=M_{\pi}$, a contradiction.

Assume thus that $i\geq k$. Then $v_i\geq v_k$. Let
%$$j^*=\begin{cases}k ,&\text{ if }i>k\\
%k-1,&\text{ if }i=k\\
%\end{cases}$$
$j^*=k$ if $i>k$ and $j^*=k-1$ if $i=k$.
Observe that $j^*\in S_{-i^*}$: indeed, it is in $S_1$ since $i\in [k]$, by the definition of $k$, and it is in $S_2$ because  $|w_{i^*}|\leq w([k]\setminus \{i^*\})$. Hence
$\hat{p}\leq \frac{|w_i| v_{j^*}}{W-|w_i|}\leq \frac{w_i v_{k}}{W-|w_i|} \le \frac{|w_i| v_{i}}{W-|w_i|}$
so $i$'s payoff is at most zero, so it has no incentive to misreport its value.
\end{proof}

\subsection{Approximation Ratio}
In this section we prove that FairInnerProduct is 5-approximate with respect to $OPT$.

\subsubsection{Optimal Continuous Canonical Laplace Mechanism}
We first characterize an individually rational, budget feasible, continuous canonical Laplace mechanism that has optimal distortion. Consider the fractional relaxation of \eqref{eq:budget-constr-prob}.
\begin{subequations}
\label{P1}
\begin{align}
\text{maximize} \quad&\sum_{i=1}^n |w_i|x_i  &\displaybreak[0]\\
\text{subject to}\quad &
\label{p1':indiv-rationality}  p_i \ge c_i(\e_i) = v_i\e_i(\bx),\quad \forall i \in [n]\\
\label{p1':budget-feasibility} & \sum_{i=1}^n p_i \le B \\
&\ 0\le x_i \le 1,\quad\forall i\in[n]
\end{align}
\end{subequations}
where
%\begin{align*}
$\e_i(\bx) = \frac{|w_i|x_i}{\sum_i |w_i|(1-x_i)}$.
%\end{align*}
A budget feasible, individually rational, (but not necessarily discrete or truthful)  canonical Laplace mechanism for the inner product has a minimal distortion among all such mechanisms if given input  $(\bv, \bw, B)$ it outputs $(\bx^*, \bp^*)$, where the latter constitute an  optimal solution to the above problem.  This characterization will yield the approximation guarantee of the DCLEF mechanism\footnote[1]{An analogous characterization of the budget-limited knapsack mechanism in \cite{budget-feasible-mechanisms} can be used to show that the mechanism is 5-approximate instead of 6-approximate.}.
%Observe that the cost, $c_i(\e_i)$ incurred by individual $i$ under a canonical Laplace mechanism  is given by $c_i(\e_i) = v_iw_ix_i/\s$. So \eqref{p1:indiv-rationality} corresponds to individual rationality. Also, \eqref{p1:budget-feasibility} corresponds to budget feasibility. Therefore, there is a one-to-one correspondence between feasible points of P1 and individually rational, budget feasible, canonical Laplace mechanisms. Further, the objective function of problem P1 is exactly the distortion of a canonical Laplace mechanisms. Thus, the optimal point of P1 corresponds to an individually rational, budget feasible canonical Laplace mechanism that minimizes distortion. Next, we will characterize the optimal point of P1.

\begin{lem}\label{lem:opt-canon-laplace}
Recall that $v_1\leq v_2\leq \ldots \leq v_n$. For  $0 \le k \le n$,  define $p(k) := \sum_{i=k+1}^n |w_i|$, if $0 \le k \le n-1$, and $p(n) := 0$. For  $0 \le k \le n$,  define $q(0) := 0$, and $q(k) := \sum_{i=1}^k v_i|w_i|$, if $1 \le k \le n$. Define %$\ell$ as
%\[
$\ell := \min\left\{k:\ \forall i > k,q(i)-B p(i)>0 \right\}$
%\]
%Let $a\in [0,1]$ such that $h(a) = \min\left\{\frac{p(k^*)}{q(k^*)}, \frac{2\sqrt{2}v_{k^*+1}}{B}\right\}$.
 and let %let $x^*_i,\ i\in [n],$ to be
\[	
x^*_i := \left\{
\begin{array}{rl}
1,& \text{ if }i \le \ell\\
\frac{Bp(\ell)-q(\ell)}{(v_{\ell+1}+B)|w_{\ell+1}|},& \text{ if }i = \ell+1\\
0,& \text{ if }i > \ell+1
\end{array},~\text{and}~p^*_i=v_i|w_i|x_i^*/\sigma(\bx^*)\quad i\in [n].
\right.
\]
Then $(\bx^*,\bp^*)$ is an optimal solution to \eqref{P1}.
\end{lem}
\begin{proof}
We show first that the quantities $\ell$ and $x_i^*$ are well defined. For $p(i)$, $q(i)$, $i\in \{0,\ldots,n\}$, as defined in the statement of the theorem, observe that $g(i)=q(i)-Bp(i)$ is strictly increasing and that $g(0)<0$ while $g(n)>0$. Hence, $\ell$ is well defined; in particular, $\ell\leq n-1$. The monotonicity of $g$ implies that $g(i)\leq 0 $ for all $0\leq i\leq \ell$ and $g(i)>0$ for $i>\ell$.
For $a\in [0,1]$, let $h(a)= q(\ell)+v_{\ell+1}|w_{\ell+1}|a -B(p(\ell+1)+|w_{\ell+1}|(1-a)).$ Then $h(0)= g(\ell)\leq 0$ and $h(1)=g(\ell+1)>0$. As $h(a)$ is continuous and strictly increasing in the reals, there exists a unique $a^*\in [0,1]$ s.t.~$h(a)=0$; since $h$ is linear, it is easy to verify that $a^*=q(\ell)-Bp(\ell)/(v_{\ell}+B)|w_{\ell+1}|=x^*_{\ell+1}$ and, hence, $x^*_{\ell+1}\in [0,1]$.
%Armed with these observations, we show that \eqref{P1} is equivalent to a convex optimization problem.
%\begin{lem}\label{lem:equiv-P1-P2}
To solve \eqref{P1},  we need only consider cases for which constraint \eqref{p1':indiv-rationality} is tight, \ie, $p_i = v_i\e_i(\bx)$.  Any solution for which \eqref{p1':indiv-rationality} is not tight can be converted to a solution where it is; this will only strengthen constraint \eqref{p1':budget-feasibility}, and will not affect the objective.
Thus, \eqref{P1} is equivalent to:
\begin{subequations}
\label{P2}
\begin{align}
\text{Max.} \quad& F(\bx)=\sum_{i=1}^n |w_i|x_i \\
\text{subj.~to}\quad &\sum_{i=1}^n  v_i|w_i|x_i - B\sum_{i=1}^n w_i(1-x_i)\leq 0, \quad
{\bx\in [0,1]^n}
\end{align}
\end{subequations}
%Observe that the objective of \eqref{P1} is given by $W - \sqrt{F(\bx)}$, which implies that maximizing the objective of \eqref{P1} is equivalent to minimizing $F(\bx)$. Given a feasible solution $\bf{x}$ of \eqref{P2}, let $p_i=v_iw_ix_i/\sigma(\bf{x})$; then  $\bf{x},\bf{p}$ is a feasible solution of $\eqref{P1}$, and the objective for both problems is $F(\bf{x})$. Similarly, given a feasible solution $\bf{x},\bf{p}$ of P1, this can always be converted to a solution $\bf{x},\bf{p}'$ where $p_i'=v_iw_ix_i/\sigma(\bf{x})$ (\emph{i.e.}, the constraints \eqref{p1':indiv-rationality} are tight), which is feasible as it strengthens constraint \eqref{p1':budget-feasibility}. Hence, $\bf{x}$ is also a feasible solution of P2. Since any solution of one problem can be converted to a solution of the other problem that evaluates to the same objective, \eqref{P1} and \eqref{P2} are equivalent.
%\end{lem}
%\begin{proof}
%\end{proof}

It thus suffices to show that $\bx^*$ is an optimal solution to \eqref{P2}. The latter is a linear program and its Lagrangian is % Therefore, KKT conditions are necessary and sufficient for optimality. %, \emph{i.e.} any $\bx\in [0,1]^n$ for which KKT conditions hold must be an optimal solution to P2.

%\begin{align*}
$$L(\bx,\l,{ \mu,\nu}) =  -F(\bx) + \lambda\big(\sum_{i=1}^n  v_i|w_i|x_i - B\sum_{i=1}^n |w_i|(1-x_i)\big)
  + \sum_{i=1}^n \mu_i(x_i-1) -\sum_{i=1}^n \nu_ix_i.$$
%\end{align*}
It is easy to verify that $\bx^*$ satisfies the KKT conditions of \eqref{P2} with
$\lambda^* = \frac{1}{v_{\ell+1}+B}$, $\mu_i^*= \id_{(i\le\ell)}\cdot \frac{v_{\ell+1}-v_i }{v_{\ell+1}+B} |w_i|$, and $\nu_i^* =\id_{(i>\ell+1)}\cdot\frac{v_i - v_{\ell+1}}{v_{\ell+1}+B} |w_i|$. 
%$$\mu_i^* := \begin{cases}
%\frac{v_{\ell+1}-v_i }{v_{\ell+1}+B} |w_i|,  &i\le\ell\\
%0, &i > \ell
%\end{cases};\quad 
%\nu_i^* :=
%\begin{cases}
%0, & i\leq \ell+1\\
%\frac{v_i - v_{\ell+1}}{v_{\ell+1}+B} |w_i|, &i> \ell+1
%\end{cases}\qed
%~~~~~~~~~~~~~~~~~~~~~~~~~~\qedhere
%$$
\end{proof}

A canonical Laplace mechanism that outputs $(\bx^*, \bp^*)$ given by Lemma~\ref{lem:opt-canon-laplace} would be optimal. Moreover, the objective value $S(\bx^*;\bw) \ge OPT$.

\begin{prop}\label{lem:l-k}
Let $\ell$ be as is defined in Lemma~\ref{lem:opt-canon-laplace}, and $k$ as defined in FairInnerProduct. Then, $\ell \ge k$.
\end{prop}
\begin{proof}
%By contradiction. 
Assume that $\ell < k$. Then
%\begin{align*}
$$
B(W - w([k])) \le B(W - w([\ell+1])) < \sum_{i=1}^{\ell+1} |w_i| v_i 
\le \sum_{i=1}^k |w_i|v_i \le v_k\sum_{i\le k} w_i = v_kw([k]).
$$
%\end{align*}
However, this contradicts the fact that $B/w(k) \ge v_k/(W - w([k]))$.
\end{proof}

\begin{prop}\label{prop:wk+1}
Let $\{x^*_i\}$ and $\ell$ be as defined in Lemma~\ref{lem:opt-canon-laplace}, and $k$ as defined in FairInnerProduct. Then, $w([k+1]) > \sum_{i=k+1}^{\ell+1} |w_i|x^*_i$.
\end{prop}
\begin{proof}
If $\ell = k$, the statement is trivially true. Consider thus the case  $\ell > k$. %The proof is by contradiction. 
Assume that $\sum_{i=1}^{k+1} w_i \le \sum_{i=k+1}^{\ell+1} |w_i|x^*_i$. Then,
\begin{align*}
\frac{B(W - w([k+1]))}{w([k+1])} &\ge \frac{B(W - \sum_{i=k+1}^{\ell+1} |w_i|x^*_i)}{\sum_{i=k+1}^{\ell+1} |w_i|x^*_i} % \text{ (by assumption)}\\
%&
\ge \frac{B(W - \sum_{i=1}^{\ell+1} |w_i|x^*_i)}{\sum_{i=k+1}^{\ell+1} |w_i|x^*_i}\displaybreak[0]\\
%&= \frac{\sum_{i=1}^{\ell+1} |w_i|v_ix^*_i}{\sum_{i=k+1}^{\ell+1} |w_i|x^*_i} %\text{ (by definition of $\ell$ and $x^*_i$'s)}\\
&
\ge \frac{\sum_{i=k+1}^{\ell+1} |w_i|v_ix^*_i}{\sum_{i=k+1}^{\ell+1} |w_i|x^*_i}\ge v_{k+1}
\end{align*}
since $v_{k+1} \le v_i$ for all $(k+1) \le i \le \ell$.
 However, this contradicts the fact that $B/w([k+1]) < v_{k+1}/(W - w([k+1]))$.
\end{proof}

Now we will show that $S(\bx;\bw) \ge \frac15 OPT$ using Proposition~\ref{prop:wk+1}. First notice that since \eqref{P1} is a relaxation of \eqref{eq:budget-constr-prob}, $OPT \le S(\bx^*;\bw)$, where $\{x^*_i\}$ are defined in Lemma~\ref{lem:opt-canon-laplace}. Therefore, we have that
%\begin{align*}
$OPT %&
\le S(\bx^*;\bw) = \sum_{i\le k} |w_i| + \sum_{i=k+1}^{\ell+1} |w_i|x^*_i %\\
 \stackrel{\text{Prop.~\ref{prop:wk+1}}}{<} w([k]) + w([k+1]) %\text{ (by Proposition~\ref{prop:wk+1})}\\
%= 2w([k]) + w_{k+1} 
\le 2w([k]) + |w_{i^*}|
$
%\end{align*}
It follows that if $O = [k]$, it implies $w([k]) \ge |w_{i^*}|$ and therefore $w([k]) =  S(\bx;\bw) \ge \frac13 OPT$. On the other hand, if $O = \{i^*\}$, then $|w_{i^*}| > \sum_{j\in [k+1]\setminus\{i^*\}} |w_j|$, which implies $2w_{i^*} > w([k])$. Therefore, $OPT \le 2w([k]) + |w_{i^*}| < 5|w_{i^*}| = 5 S(\bx; \bw)$.%\qed

\subsection{The Uniform-Weight Case}
In this section, we prove that when all weights are equal, FairInnerProduct is 2-approximate with respect to $OPT$.

Let $|w_i| = u$ for all $i\in [n]$. First, observe that in this case, FairInnerProduct always outputs $O =[k]$. Therefore, $S(\bx; \bw) = ku$. We use this observation to prove the result.

\begin{lem}
\label{lem:2k>l-1}
Assume that for all $i\in [n],\ |w_i| = u$. Then, $S(\bx; \bw) \ge \frac12 OPT$.
\end{lem}
\begin{proof}
Observe that
%\begin{align*}
$OPT \le S(\bx^*; \bw) =  \sum_{i=1}^{\ell+1} |w_i| x^*_i =  w([k]) + \sum_{i=k+1}^{\ell+1} |w_i| x^*_i 
< w([k]) + w([k+1])$ %\text{ (from Proposition~\ref{prop:wk+1})}
%\end{align*}
from Proposition~\ref{prop:wk+1},
where $\{x^*_i\}$ and $\ell$ are defined in Lemma~\ref{lem:opt-canon-laplace}. Substituting $|w_i| = u$ for all $i$, we get
%\[
$OPT < (2k+1)u$
%\]
Since $OPT$ is the objective value attained by the optimal DCLEF mechanism, $OPT = mu$ for some $m\in [n]$.
This implies $2k + 1> m$.  Since $k$ and $m$ are integers, it follows that $2k \ge m$, or equivalently, $S(\bx; \bw) \ge \frac12 OPT$.
\end{proof}

\section{Proof of Theorem~\protect\lowercase{\ref{thm:hardness}} (Hardness of Approximation)}\label{sec:hardness}
Consider the following example. Let $n = 4$. The private costs of the four individuals are given by $v_1 = a, v_2 = v_3 = v_4 = 2$, where $0 < a < 2$. The weights of the four individuals are given by $w_1 = w_2 = w_3 = w = d$, where $d > 0$. Let the budget $B = 1 + a/2 < 2$.

Observe that the optimal individually rational, budget-feasible, DCLEF  mechanism would set $x^*_1 = 1$ and exactly one of $x^*_2, x^*_3$ and $x^*_4$ to 1. 
Without loss of generality, assume that $x^*_1 = x^*_2 = 1$ and $x^*_3 = x^*_4 = 0$. 
Therefore, the optimal weight $OPT = 2d$.  
Consider a \emph{truthful} DCLEF mechanism that is $2-\ve$ approximate, for any $\ve > 0$. 
Such a mechanism must set $x_1 = 1$ (since it is truthful) and at least one more $x_i$ to 1 (since it is $2-\ve$ approximate). 
Therefore, for such a mechanism $\s(\bx) \le 2d$. 
This implies that for such a mechanism, the cost of individual 1, $c_1(\e_1) = v_1w_1/\s(\bx) \ge v_1d/(2d) \ge v_1/2$. 
Since the mechanism is truthful, the payment $p_1$ cannot depend on $v_1$. 
Also, for this mechanism to be individually rational, $p_1$ must be at least 1 (since $v_1$ can be arbitrarily close to 2), which implies that the remaining budget is strictly less than 1. 
However, for this mechanism, for $i\in \{2,3,4\},\ c_i(\e_i) = 2d/\s(\bx) \ge 1$. 
This means that this mechanism cannot be both individually rational and budget feasible. \qed

%\appendixhead{DANDEKAR}

% Acknowledgments
%\begin{acks}
%The authors would like to thank Dr. Maura Turolla of Telecom
%Italia for providing specifications about the application scenario.
%\end{acks}

% Bibliography
%\bibliographystyle{acmsmall}
%\bibliography{recsys_privacy_bib}

% History dates
%\received{February 2007}{March 2009}{June 2009}

% Electronic Appendix
%\elecappendix

\end{document}